\documentclass[12pt]{iopart}
\usepackage{amsfonts}
\usepackage[dvips]{color}
\usepackage[dvips]{graphicx}
\usepackage{latexsym}
\usepackage{amsthm}	

\newtheorem{theo}{Theorem}
\newtheorem{lemma}[theo]{Lemma}
\newtheorem{prop}[theo]{Proposition}
\newtheorem{definition}[theo]{Definition}

\begin{document}
\title[Nonlinear Lattice with Potential Symmetry for Smooth Propagation of DB]{Construction of Nonlinear Lattice\\ with Potential Symmetry\\
 for Smooth Propagation of Discrete Breather}
\author{Yusuke Doi}
\address{Department of Adaptive Machine Systems, Graduate School of
Engineering, Osaka University\\
2-1 Yamadaoka, Suita, Osaka 565-0871, Japan}
\ead{doi@ams.eng.osaka-u.ac.jp}
\author{Kazuyuki Yoshimura}
\address{Faculty of Engineering, Tottori University\\
4-101 Koyama-Minami, Tottori 680-8552, Japan}
\ead{kazuyuki@tottori-u.ac.jp}
\begin{abstract}
   We construct a nonlinear lattice
 that has a particular symmetry in its potential function consisting of long-range pairwise interactions.
 The symmetry enhances smooth propagation of discrete breathers,
 and it is defined by an 
 invariance of the potential function with respect to a map acting on the complex normal mode coordinates.
 Condition of the symmetry is given by
 a set of algebraic equations with respect to
 coefficients of the pairwise interactions.
 We prove that
 the set of algebraic equations has a unique solution,
 and moreover we solve it explicitly.
 We present an explicit Hamiltonian for the symmetric lattice,
 which has coefficients given by the solution.
 We demonstrate that the present symmetric lattice is useful for
 numerically computing traveling discrete breathers in various lattices.
 We propose an algorithm using it.
\end{abstract}

\maketitle

%
%
%
%
\section{Introduction}
   Wave propagation is one of the fundamental phenomena in physics.
 It has been of crucial importance to understand
 the nature of wave propagation in spatially discrete media.
 For example,
 phonons are plane waves in crystals,
 and their characteristics dominate
 various material properties such as thermal transport.
 In addition to the discreteness,
 a variety of discrete media are inherently nonlinear.
 Therefore,
 wave propagation in nonlinear discrete media
 has been of increasing importance.

   {\it Discrete breathers} (DBs) are
 space-localized modes
 that ubiquitously emerge in a variety of nonlinear discrete media.
 The concept of DB was introduced
 by Takeno and Sievers \cite{Sievers1988},
 and it has been of great interest \cite{flach2008a,yoshimura2015}.
 Two types of DBs are known to be possible, i.e.,
 stationary and traveling DBs.
 Long-lived traveling DBs
 have been found numerically in various nonlinear lattice models
 \cite{burlakov1990,takeno1991,hori1992,sandusky1992}:
 they propagate along the lattices
 without noticeable decay for a long time.
 Such traveling DBs are of considerable interest
 from the viewpoint of energy transport,
 and their properties have been investigated
 \cite{aubry1998,aubry2006,Yoshimura2007,gomez-gardenes2004b}.

   Discreteness effects 
 manifest in propagation property of traveling DB.
 The lattice discreteness in general tends to reduce the mobility of DB:
 for instance,
 an approximate traveling DB
 produced by perturbing a stationary DB
 loses its velocity during its propagation,
 and it is eventually trapped at a certain lattice site
 \cite{sandusky1992}.
 On the other hand, it is possible to precisely compute
 a traveling DB solution without velocity loss
 by combining the Newton method
 with numerical integration of the equations of motion.
 A remarkable feature is that
 it does not propagate with a constant velocity
 but with periodically varying velocity, i.e., the non-smooth propagation
 \cite{aubry2006,doi2016}.
 The period of this velocity variation is just
 a time needed for propagating one lattice space. 
 Both the velocity loss of an approximate traveling DB
 and the velocity variation of a precise traveling DB
 vanish in the continuum limit,
 where the DB is very weakly localized \cite{Yoshimura2007,Feng2004}.
 This fact indicates that
 the two features are just manifestations of 
 the lattice discreteness effects.

   A fundamental issue is to clarify
 the origin of such discreteness effects,
 in other words,
 an essential property of the lattice potential
 that causes the discreteness effects. 
 Addressing this issue,
 we pointed out the relevance of a particular symmetry of the lattice potential
 \cite{Yoshimura2007}. 
 Recently,
 we have proposed a nonlinear lattice having this symmetry,
 which has a potential function consisting of pairwise long-range interactions. 
 We have numerically demonstrated that
 this lattice allows
 constant-velocity traveling DBs
 and moreover exhibits a high mobility of approximate traveling DBs,
 i.e.,
 the lattice is
 free from the discreteness effect \cite{doi2016}. 
 In contrast,
 it is possible to break the symmetry by adding a perturbation
 to the lattice potential.
 As the perturbation increases,
 the velocity variation of precise traveling DB becomes larger
 (cf.~Sec.~5)
 and the velocity loss of approximate traveling DB also increases 
 \cite{doi2016}.
 These results indicate that 
 breaking the symmetry is the origin of the discreteness effects
 in propagation of DBs.

  It also should be emphasized that
 the ballistic thermal transport has been numerically observed
 in lattices with the potential symmetry
 \cite{Bagchi2017,yoshimura2019, yoshimura2019a}.
 In addition,
 when the symmetry is broken by adding a perturbation of the potential,
 the lattice exhibits 
 a transition from the ballistic to a non-ballistic (but still anomalous)
 thermal transport as the lattice size increases \cite{yoshimura2019}.
 The threshold lattice size $N_c$ depends on the magnitude of the perturbation:
 $N_c$ decreases as the perturbation is increased.
 These observations indicate that
 the thermal resistance appears and becomes stronger
 as the symmetry breaks gradually.
 Thus, it is expected that
 a study of the thermal transport by gradually breaking the symmetry
 in the present lattice
 will lead to a better understanding of
 the mechanism of thermal resistance.
 The present symmetric lattice may be of much significance
 also from the point of view of thermal transport.

   In our previous paper \cite{doi2016},
 we stated that
 there exists such a symmetric lattice without a proof.
 In the present paper, we give a proof of its existence.
 Moreover,
 we give an explicit Hamiltonian for the symmetric lattice.
 Analytical expressions for the coefficients
 appearing in the symmetric lattice's potential were not given,
 but they were only numerically obtained in \cite{doi2016}.
 Here, we obtain the coefficients explicitly.

   The symmetric lattice is useful
 for numerically computing traveling DB solutions. 
 Precise computation of traveling DB usually employs the Newton method,
 which needs a good approximate solution as an initial guess.
 This approximation is obtained by perturbing a stationary DB,
 but this method does not always
 work 
 successfully
 because the domain of convergence of the Newton method is very small
 for ordinary nonlinear lattices such as the Fermi-Pasta-Ulam (FPU) lattice
 and the approximation is not precise enough.
 A useful feature of the symmetric lattice is that
 it is possible to find a precise enough approximation by perturbation.
 Given a lattice model to compute a traveling DB such as the FPU one,
 our idea for the algorithm is
 to introduce a lattice that has a potential function
 parameterized between the symmetric lattice and the given one,
 compute a precise traveling DB for the symmetric lattice case,
 and then continue it to the given lattice case
 by gradually changing the parameter value.
 This idea has already been proposed in \cite{Yoshimura2005,Doi2009}.
 However,
 only the four-particle symmetric lattice was constructed at that point,
 and the $N$-particle one has been lacking.
 The present lattice is an extension of the four-particle symmetric lattice
 to an arbitrary degrees of freedom.
 We demonstrate that
 our algorithm using the present lattice
 successfully works for computing traveling DBs
 in the FPU lattice.
 
 The rest of paper is organized as follows.
 In section \ref{sec:definition},
 the definition of the symmetric lattice is given.
 In section \ref{sec:model_and_main_result},
 we give a pairwise interaction symmetric lattice,
 as well as the main results on the existence and uniqueness of the proposed model.
 The explicit expression of the proposed model is also given.
 In section \ref{sec:calculation_method},
 the numerical method for finding traveling DBs using the proposed model is presented.
 In section \ref{sec:truncated}, we discuss physical effects of breaking the symmetry.
 In section \ref{sec:preliminary},
 we give a preliminary discussion of the proof.
 In section \ref{sec:odd_equation} and \ref{sec:even_equation},
 several lemmas are proved to prepare the proof of main result.
 The proof of main result is given in section \ref{sec:proof}.
%
%
%
\section{Definition of symmetric lattice}\label{sec:definition}

  Let us consider a nonlinear lattice described by the Hamiltonian
\begin{eqnarray}
 H = \frac{1}{2}\sum_{n=1}^{N}p_n^2+ \Phi(q_1, q_2, \ldots, q_N),
 \label{eqn:Hamiltonian_in_q}
\end{eqnarray}
 where $p_n\in\mathbb{R}$ and $q_n\in\mathbb{R}$ represent
 the linear momentum and the position of $n$th particle, respectively,
 $N$ is the number of particles of the system,
 and $\Phi(q_1,q_2, \ldots, q_N)$: $\mathbb{R}^N \to\mathbb{R}$
 is a $C^2$ function of $(q_1, q_2, \ldots, q_N)$.
 We assume the case of even $N$ and the periodic boundary conditions
 $p_{N+1} = p_{1}$, $p_{0}= p_{N}$, $q_{N+1} = q_{1}$ and $q_{0} =q_{N}$. 

   Consider the variable transformation defined by
\begin{eqnarray}
 q_n = \frac{(-1)^n}{\sqrt{N}}\sum_{m=-N/2+1}^{N/2}{U_m \exp\left[{-i\frac{2\pi n}{N}m}\right]},
 \quad n=1,2,\cdots, N,
 \label{eqn:variable_transformation_q_to_U}
\end{eqnarray}
 where $U_m \in \mathbb{C}$, $m=-N/2+1, -N/2+2, \ldots, N/2$ are called
 the complex normal mode coordinates.
 Note that the $N/2$th mode represents the uniform displacement of the lattice. 
 Substituting Eq.~(\ref{eqn:variable_transformation_q_to_U}),
 Hamiltonian (\ref{eqn:Hamiltonian_in_q}) can be written
 in terms of $U_m$ and reads
\begin{eqnarray}
 H = \frac{1}{2}\sum_{m=-N_h}^{N_h+1}{\dot{U}_m}{\dot{U}_{-m}} +
  \Phi(\mathbf{U}, U_{N/2}),
  \label{eqn:Hamiltonian_in_U}
\end{eqnarray}
 where $N_h = N/2-1$ and $\mathbf{U} = (U_{-N_h}, U_{-N_h+1}, \ldots, U_{N_h})$.
 The potential $\Phi (\mathbf{U},U_{N/2})$ can be decomposed as
\begin{equation}
  \Phi(\mathbf{U},U_{N/2}) = \Phi(\mathbf{U},0)+\mathcal{G}(\mathbf{U},U_{N/2}),
  \label{eqn:UR}
\end{equation}
 where $\mathcal{G}(\mathbf{U}, U_{N/2})
 \equiv \Phi(\mathbf{U},U_{N/2})-\Phi(\mathbf{U},0)$.

   Consider the map $\mathcal{T}_\lambda:\mathbb{C}^{N-1} \to \mathbb{C}^{N-1}$
 defined by
\begin{eqnarray}
 \mathcal{T}_\lambda : U_m \mapsto U_m \exp(-i m\lambda),
 \quad m=-N_h,\dots,N_h,
 \label{eqn:transform_T}
\end{eqnarray}
 where $\lambda$ is a real parameter.
 This map $\mathcal{T}_\lambda$ forms a one-parameter transformation group. 
  When $\lambda = 2\pi/N$,
 given an arbitrary displacement pattern $\mathbf{q}=(q_1,q_2,\cdots, q_{N})$
 satisfying $U_{N/2}= \sum_{n=1}^{N}q_n/\sqrt{N}=0$, 
 the map $\mathcal{T}_\lambda$ represents successive operations of 
 shifting $\mathbf{q}$ by one lattice spacing
 and reversing the sign of the resulting displacement.
 Therefore,
 $\mathcal{T}_\lambda$ with an arbitrary $\lambda$ may be regarded as
 a continuous extension of this one-lattice-space shifting
 and sign-inverting transformation.

   The $U_{N/2}$-independent part of potential function $\Phi(\mathbf{U},0)$
 in Eq.~(\ref{eqn:UR}) can be divided into two parts:
 $\Phi_\mathrm{s}(\mathbf{U})$ and $\Phi_\mathrm{a}(\mathbf{U})$.
 The former part is invariant with respect to the map $\mathcal{T}_\lambda$
 for any $\lambda\in\mathbb{R}$ and any $\mathbf{U}\in\mathbb{C}^{N-1}$,
 i.e.,
 $\Phi_\mathrm{s}(\mathcal{T}_\lambda\mathbf{U})=\Phi_\mathrm{s}(\mathbf{U})$,
 while 
 $\Phi_\mathrm{a}(\mathbf{U}) = \Phi(\mathbf{U},0)-\Phi_\mathrm{s}(\mathbf{U})$
 is the rest part of $\Phi(\mathbf{U},0)$
 and not invariant with respect to $\mathcal{T}_\lambda$ for any $\lambda$.
 We call the former part $\Phi_\mathrm{s}(\mathbf{U})$ the symmetric part.
 This decomposition in Eq.(\ref{eqn:Hamiltonian_in_U}) provides
\begin{equation}
  H = \frac{1}{2}\sum_{m=-N_{h}}^{N_h + 1}{\dot{U}_m\dot{U}_{-m}}
  +\Phi_\mathrm{s}(\mathbf{U})+\Psi(\mathbf{U}, U_{N/2}),
  \label{eqn:hamiltonian_decomposed}
\end{equation}
 where $\Psi = \Phi_\mathrm{a}(\mathbf{U}) + \mathcal{G}(\mathbf{U},U_{N/2})$
 is the asymmetric part of the whole potential $\Phi(\mathbf{U},U_{N/2})$.
   Let $\mathcal{I} =\{(\mathbf{q},\mathbf{p})\in \mathbb{R}^{2N}|\,
 \sum_{n=1}^{N}q_n=\sum_{n=1}^{N}p_n =0\}$.
 This is the subspace which is specified by $U_{N/2}=\dot{U}_{N/2}=0$.
 We give the following definition.
\begin{definition}
 The lattice (\ref{eqn:Hamiltonian_in_U}) or (\ref{eqn:hamiltonian_decomposed})
 is said to be a symmetric lattice
 if $\mathcal{I}$ is an invariant subspace
 and $\Psi(\mathbf{U},0)=0$.
\label{def:symmetric}
\end{definition} 
 By this definition,
 the Hamiltonian of the reduced dynamical system on $\mathcal{I}$
 of a symmetric lattice is given in the form
\begin{eqnarray}
 H &=&
  \frac{1}{2}\sum_{m=-N_h}^{N_h}{\dot{U}_m\dot{U}_{-m}} +
  \Phi_\mathrm{s}(\mathbf{U}).
  \label{eqn:symmetric_Hamiltonian_in_U}
\end{eqnarray}
 It has been reported that the following two propositions hold in the symmetric lattice\cite{Doi2009}.
 \begin{prop}
 Suppose that the symmetric lattice (\ref{eqn:symmetric_Hamiltonian_in_U}) has a solution $U_m(t)=u_m(t), m=-N_h,\dots,N_h, U_{N/2}(t)=0$. Then for any $\lambda \in \mathbb{R}$, $U_m(t) = u_m(t)\exp{(-im\lambda)}, m=-N_h,\dots,N_h, U_{N/2}(t)=0$ is also a solution.	
 \end{prop}
 
 \begin{prop}
 	The symmetric lattice (\ref{eqn:symmetric_Hamiltonian_in_U}) has an additional first integral given by
 	\label{prop:extra_conserve}
    \begin{equation}
	   I = \sum_{m=1}^{N/2-1}{m(\dot{U}_mU_{-m}-U_m\dot{U}_{-m})}.
	   \label{eqn:extra_conserved_quantity}
    \end{equation}
 \end{prop}
 
%
%
\section{Model and main result}\label{sec:model_and_main_result}
   It is possible to construct a symmetric lattice from
 any lattice system defined by Hamiltonian (\ref{eqn:Hamiltonian_in_U}),
 because it is enough to eliminate its asymmetric terms.
 However,
 such a model is in general unphysical
 when it is transformed into the Hamiltonian in terms of $q_n$.
 For example,
 such lattice is not composed of pairwise interactions.
 Our purpose is to construct a symmetric lattice
 that has a potential consisting of pairwise interactions only.

 In section \ref{sec:definition}, we considered the symmetric lattice defined for even $N$.
 For simplicity of the proof, we restrict the following discussion to the case that $N$ is a multiple of 4, i.e, $N/2$ is even. We have not given the proof for the case that $N/2$ is odd. However, it may be possible to prove as the same manner.
   Let us consider the Hamiltonian which has pairwise interaction terms as follows:
\begin{eqnarray}
 H =\sum_{n=1}^{N}{\frac{1}{2} p_{n}^2}+\sum_{n=1}^{N}{\left[\frac{\mu_0}{2}q_n^2+\frac{\mu_1}{2}(q_{n+1}-q_n)^2\right]} +
  \frac{1}{4}\sum_{n=1}^{N}{\sum_{r=1}^{N/2}{b_r (q_{n+r}-q_n)^4}},\nonumber\\
 \label{eqn:proposed_model}
\end{eqnarray}
 where $b_r \in \mathbb{R}$ is a constant
 which represents the coupling strength between the $r$th neighboring particles,
 $\mu_0$ and $\mu_1$ are the coefficients of harmonic on-site potential and harmonic intersite potential, respectively.
 Hamiltonian (\ref{eqn:proposed_model}) reduces to
 the FPU-type lattice when $\mu_0 =0$ and $b_r = 0\quad (r\ge2)$.

   Substituting Eq.~(\ref{eqn:variable_transformation_q_to_U})
 into Eq.~(\ref{eqn:proposed_model}),
 we rewrite the Hamiltonian in terms of $U_m$ into the form
\begin{eqnarray}
 H =
  \frac{1}{2}\sum_{m=-N_h}^{N_h+1}{\dot{U}_m\dot{U}_{-m}}
  +\Phi_\mathrm{s}(\mathbf{U})+\Phi_\mathrm{a}(\mathbf{U}) + \frac{\mu_0}{2}U_{N/2}^2.
  \label{eqn:hamiltonian_in_U_model}
\end{eqnarray}
It is easy to see that Hamiltonian (\ref{eqn:hamiltonian_in_U_model}) has the invariant subspace $\mathcal{I}$.  
Comparing Eq.(\ref{eqn:hamiltonian_in_U_model}) with Eq.(\ref{eqn:hamiltonian_decomposed}), we see $\Psi(\mathbf{U}, U_{N/2})=\Phi_\mathrm{a}(\mathbf{U}) + (\mu_0/2)U_{N/2}^2$. 
This asymmetric part reduces $\Psi(\mathbf{U}, 0)=\Phi_\mathrm{a}(\mathbf{U})$ on $\mathcal{I}$.
In order to construct a symmetric lattice, it is enough to consider $\Phi_\mathrm{a}(\mathbf{U})$ as the asymmetric part.
 The symmetric part and asymmetric part in the lattice potential are given as follows:
\begin{eqnarray}
\Phi_\mathrm{s}(\mathbf{U})
&=&\frac{1}{2}\sum_{m=-N_h}^{N_h}{\left[\mu_0 +4\mu_1 \cos^2\left(\frac{\pi m}{N}\right)\right]}U_m U_{-m}\nonumber\\
&+&\frac{4}{N}\sum_{i,j,k,l=-N_h}^{N_h}
{\phi^{(i,j,k,l)}(\mathbf{b}) U_i U_j U_k U_l}\Delta(i+j+k+l)
\label{eqn:potential_symmetric}
\end{eqnarray}
\begin{eqnarray}
 \Phi_\mathrm{a}(\mathbf{U})
  &=&-\frac{4}{N}\sum_{i,j,k,l=-N_h}^{N_h}
{\psi^{(i,j,k,l)}(\mathbf{b}) U_i U_j U_k U_l}\Delta(i+j+k+l+N)\nonumber\\
  &&-\frac{4}{N}\sum_{i,j,k,l=-N_h}^{N_h}
{\psi^{(i,j,k,l)}(\mathbf{b}) U_i U_j U_k U_l}\Delta(i+j+k+l-N)
\label{eqn:potential_asymmetric}
\end{eqnarray}
with
\begin{eqnarray}
\phi^{(i,j,k,l)}(\mathbf{b})&=&\sum_{q=1}^{N/2}{b_q f_{q}^{(i,j,k,l)}}
\label{eqn:coefficient_symmetric_term}
\end{eqnarray}
\begin{eqnarray}
\psi^{(i,j,k,l)}(\mathbf{b})&=&-\sum_{q=1}^{N/2}(-1)^{q}{b_q f_{q}^{(i,j,k,l)}},
\label{eqn:coefficient_asymmetric_term}
\end{eqnarray}
 where $\mathbf{b}=[b_1, b_2, \ldots, b_{N/2}]^{T}$ and
\begin{eqnarray}
 f_{q}^{(i,j,k,l)} =
  \left\{
   \begin{array}{ll}
    c_{iq}c_{jq}c_{kq}c_{lq}&\mbox{for
     odd $q$}\\
    s_{iq}s_{jq}s_{kq}s_{lq}&\mbox{for
     even $q$,}
   \end{array}
  \right.
\end{eqnarray}
 where $c_\alpha = \cos(\alpha\pi/N)$ and $s_\alpha = \sin(\alpha\pi/N)$.
 The function $\Delta(d)$ is defined by
\begin{eqnarray}
  \Delta(d) =\left\{\begin{array}{ll}
	      1&\mbox{if}\quad d=0\\
		     0&\mbox{otherwise.}
		    \end{array}\right.
\end{eqnarray}

   Let $\mathcal{S}_{\pm} = \{(i,j,k,l)\in\mathbb{Z}^{4}|-N_h \le i,j,k,l\le N_h,
 i+j+k+l=\pm N \}$.
 The lattice (\ref{eqn:hamiltonian_in_U_model}) becomes symmetric
 if and only if the asymmetric part (\ref{eqn:potential_asymmetric}) vanishes.
 The condition $\Phi_\mathrm{a}(\mathbf{U})=0$ is equivalent to
\begin{eqnarray}
 \psi^{(i,j,k,l)}(\mathbf{b})=0,\quad \forall (i,j,k,l)\in \mathcal{S}_{\pm}.
 \label{eqn:equation_for_symmetric}
\end{eqnarray}
 As for Eq.~(\ref{eqn:equation_for_symmetric}),
 we can readily obtain the following lemma. 
\begin{lemma}
 Suppose that
 $\psi^{(i,j,k,l)}(\mathbf{b})=0$ holds for $(i,j,k,l)\in\mathbb{Z}^4$
 and $\mathbf{b}\in\mathbb{R}^{N/2}$,
 then we have $\psi^{(-i,-j,-k,-l)}(\mathbf{b})=0$.
\label{lemma:minus}
\end{lemma}
\begin{proof}
 For any $q\in\mathbb{Z}$,
 $s_{-(2q-1)i}=-s_{(2q-1)i}$ and $c_{-2qi} = c_{2qi}$ hold.
 This fact implies $f_{q}^{(-i,-j,-k,-l)}=f_{q}^{(i,j,k,l)}$.
 Therefore, we have
 $\psi^{(-i,-j,-k,-l)}(\mathbf{b})=\psi^{(i,j,k,l)}(\mathbf{b})$.	
\end{proof}
 Since Eq.(\ref{eqn:equation_for_symmetric}) is invariant
 under any exchange of indices,
 we can restrict our discussion to
 the subset $S_0=\{(i,j,k,l)|-N_h \le i \le j \le k \le l \le N_h, i+j+k+l=\pm N\}
 \subset \mathcal{S}_{\pm}$ without loss of generality.
 Using this fact and Lemma \ref{lemma:minus},
 we can further restrict our discussions to the set
 $S = \{(i,j,k,l)|-N_h \le i \le j \le k \le l \le N_h, i+j+k+l= N\}\subset S_0$.
 Finally, it is enough to discuss the equations in the set $S$ instead of Eq.~(\ref{eqn:equation_for_symmetric}) which is in the set $\mathcal{S}_{\pm}$.
 Therefore, we consider the equations
 \begin{eqnarray}
	\psi^{(i,j,k,l)}(\mathbf{b}) = 0, \forall (i,j,k,l)\in S.
	\label{eqn:equation_for_s}
 \end{eqnarray}

   Let $S_1=\{(0,n+1,N/2-n, N/2-1)|1\le n \le N/4-1\}\subset S$
 and $S_2=\{(2-m,m,N/2-1,N/2-1)|m=1\ \mbox{or}\ 3\le m \le N/4+1\}\subset S$.
 The following lemma holds:
\begin{lemma}
 Let $N\in\mathbb{N}$ be a multiple of 4.
 Equations $\psi^{(i,j,k,l)}(\mathbf{b})=0, \forall (i,j,k,l)\in S_1\cup S_2$
 are $N/2-1$ linearly independent equations
 and therefore they have a nontrivial solution $\mathbf{b}\ne 0$.
 Moreover, this nontrivial solution $\mathbf{b}$ also solves
 the other equations in $S$.
\label{theo:main1}
\end{lemma}
 We briefly describe the procedure of our proof of Lemma \ref{theo:main1} below:
\begin{enumerate}
	\item Showing that $N/4-1$ equations $\psi^{(i,j,k,l)}(\mathbf{b})=0, \forall(i,j,k,l)\in S_1$ are linearly independent by showing the matrix rank of this set of equations is $N/4-1$;
	\item Showing that $N/4$ equations $\psi^{(i,j,k,l)}(\mathbf{b})=0, \forall(i,j,k,l)\in S_2$ are linearly independent provided that (i) is satisfied, by showing the matrix rank of this set of equations is $N/4$;
	\item Showing that the nontrivial solution obtained by (i) and (ii) solves the other all equations $\psi^{(i,j,k,l)}(\mathbf{b})=0, \forall(i,j,k,l)\in S\setminus(S_1\cup S_2)$.
\end{enumerate}
The proof of this lemma will be given in Secs.~\ref{sec:odd_equation} and \ref{sec:even_equation}.

 The solution $\mathbf{b}$ of the set of equations in Lemma ~\ref{theo:main1}
 can also be explicitly obtained as in Lemma 6. Its detailed derivation will be given in \ref{sec:explicit_solution}.

\begin{lemma}
Let $N\in\mathbb{N}$ be a multiple of 4 and 
 $b_1$ be a nonzero constant.
 The nontrivial solution of equations
 $\psi^{(i,j,k,l)}(\mathbf{b})=0, \forall (i,j,k,l)\in S_1\cup S_2$ is given by
\begin{eqnarray}
 b_r = \left\{
 \begin{array}{l}
	\displaystyle\frac{b_1\sin^2{\frac{\pi}{N}}}{\sin^2{\frac{r\pi}{N}}}\quad (r=1,2,\cdots, N/2-1),\\
	\displaystyle\frac{b_1}{2}\sin^2{\frac{\pi}{N}}\quad(r=\frac{N}{2}).
 \end{array}\right.
 \label{eqn:explicit_coeffient}
\end{eqnarray}
\label{theo:solution1}
\end{lemma}

   Combining Lemmas 5 and 6, we can obtain the following main theorem.
   The lattice model given in the following main theorem is called {\it the pairwise interaction symmetric lattice} (PISL)\cite{doi2016}.
\begin{theo}
 Let $N\in\mathbb{N}$ be a multiple of 4
 and $\mathbf{b} = (b_1, b_2, \cdots, b_{N/2})$ be
 a nontrivial solution of equations
 $\psi^{(i,j,k,l)}(\mathbf{b})=0, \forall (i,j,k,l)\in S_1 \cup S_2$. 
 Then, the lattice defined by the Hamiltonian (\ref{eqn:proposed_model})
 is a symmetric lattice, and  the explicit expression of Hamiltonian (\ref{eqn:proposed_model}) is given as follows:
%

%
\begin{eqnarray}
 H &=& \frac{1}{2}\sum_{n=1}^{N}{p_{n}^2}+\frac{1}{2}\sum_{n=1}^{N}{\left[\mu_0 q_n^2+\mu_1(q_{n+1}-q_n)^2\right]}\nonumber\\
  &+&\frac{1}{4}\sum_{n=1}^{N}{\sum_{r=1}^{N/2-1}{\frac{b_1\sin^2{\frac{\pi}{N}}}{\sin^2{\frac{r\pi}{N}}}
  (q_{n+r}-q_n)^4}}+  \frac{1}{8}\sum_{n=1}^{N}{b_1\sin^2{\left(\frac{\pi}{N}\right)} (q_{n+N/2}-q_n)^4},\nonumber\\
\label{eqn:explicit_symmetric_lattice}	
\end{eqnarray}
 where $b_1$ is an arbitrary nonzero constant.
\label{theo:main2}
\end{theo}

%
%
\section{Calculation method for traveling DB}\label{sec:calculation_method}

  It has been reported in our previous work\cite{doi2016} that
 one of the good features of the PISL is that DBs in the PISL have smooth mobility,
 that is, each traveling DB propagates with a constant velocity.
 It has been also reported that
 the PISL has a rather large tolerance in the initial perturbation
 that generates a smoothly propagating traveling DB from a stationary DB.
 These results indicate the proposed model is useful
 for obtaining an initial guess for finding a traveling DB by iteration method.
 We propose the following procedure (\ref{enm:step1})-(\ref{enm:step7}) for computing
 a traveling DB solution, which utilizes the PISL.
\begin{enumerate}
\item Consider the following lattice
 which has a parameter $C$ controlling the symmetry of lattice
 (PISL for $C=1$ and FPU-$\beta$ for $C=0$).
\begin{eqnarray}
 H &=& \sum_{n=1}^{N}\left[\frac{1}{2}{p_{n}^2} +
 \frac{1}{2}(q_{n+1}-q_n)^2+\frac{b_1}{4} (q_{n+1}-q_n)^4\right]
\nonumber\\
 &&+\frac{C}{4}\sum_{n=1}^{N}{\sum_{r=2}^{N/2}{b_r (q_{n+r}-q_n)^4}},
\label{eqn:explicit_tascl}	
\end{eqnarray}
 where $b_r$ is given by Eq.~(\ref{eqn:explicit_coeffient}).
 This lattice was named
 the {\it translational asymmetry controlled lattice (TASCL)} \cite{Doi2009}.
 The equations of motion are given by
\begin{eqnarray}
 \dot{q}_n &=& p_n,
\label{eqn:motion_tascl_q}
\\
 \dot{p}_n &=& q_{n+1}+q_{n-1}-2q_n + b_1\left[(q_{n+1}-q_{n})^3+(q_{n-1}-q_n)^3\right]
\nonumber
\\
 &&+C\sum_{r=2}^{N/2}{b_r\left[(q_{n+r}-q_n)^3+(q_{n-r}-q_n)^3\right]},
\label{eqn:motion_tascl}
\end{eqnarray}
 where $n=1,2,\cdots,N$. 
 Denote $(\mathbf{q},\mathbf{p})=(q_1,\cdots,q_N,p_1,\cdots,p_N)$. 
 The temporal evolution of a solution
 with its initial condition $(\mathbf{q}(0),\mathbf{p}(0))$
 is obtained by integrating
 Eqs.~(\ref{eqn:motion_tascl_q}) and (\ref{eqn:motion_tascl}).
 This temporal evolution over the duration $\tau$
 induces the map
 $\mathcal{F}_{C,\tau}:\mathbb{R}^{2N}\to\mathbb{R}^{2N}$
 as follows:
\begin{eqnarray}
 \mathcal{F}_{C,\tau}(\mathbf{q}(0),\mathbf{p}(0))
 =(\mathbf{q}(\tau),\mathbf{p}(\tau)).
\label{eqn:evolution_map}
\end{eqnarray}
\label{enm:step1}
%

\item Construct an approximate stationary DB
 with a prescribed angular frequency $\omega_\mathrm{DB}$
 in the symmetric lattice ($C=1$).
 Assume the approximate stationary DB solution in the form
\begin{eqnarray}
 q_n(t) = a_n \cos{\omega_\mathrm{DB} t},\quad n=1,2,\cdots,N,
\label{eqn:stationary_db}
\end{eqnarray}
 where $a_n$ represents the spatial profile of stationary DB.
 Substituting Eq.~(\ref{eqn:stationary_db})
 into Eqs.~(\ref{eqn:motion_tascl_q}) and (\ref{eqn:motion_tascl})
 and performing the rotating wave approximation,
 we obtain the algebraic equations for $a_n$ as follows:
\begin{eqnarray}
 &&a_{n+1}+a_{n-1}-(2-\omega_\mathrm{DB}^2)a_n
 + \frac{3b_1}{4}\left[(a_{n+1}-a_{n})^3+(a_{n-1}-a_n)^3\right]
\nonumber
\\
 &&+\frac{3}{4}\sum_{r=2}^{N/2}{b_r\left[(a_{n+r}-a_n)^3
 +(a_{n-r}-a_n)^3\right]}=0,
 \quad n=1,2,\cdots,N.\quad
\label{eqn:equation_an_sym}
\end{eqnarray}
 The particle amplitudes $a_n,\,n=1,2,\cdots N$,
 of the approximate stationary DB
 are obtained by numerically solving Eq.(\ref{eqn:equation_an_sym}).
\label{enm:step2}

\item Construct a precise numerical solution of stationary DB
 with the angular frequency $\omega_\mathrm{DB}$
 in the symmetric lattice ($C=1$)
 under the constraint $U_{N/2}=0$.
 This is performed by finding the periodic orbit in the phase space.
 Let $(\mathbf{q}^{(\mathrm{s})}(0),\mathbf{p}^{(\mathrm{s})}(0))$ be
 the initial state of the stationary DB
 and $T = 2\pi/\omega_\mathrm{DB}$ be
 its internal oscillation period.
 The temporal evolution map
 $\mathcal{F}_{1,T}(\mathbf{q}^{(\mathrm{s})}(0),\mathbf{p}^{(\mathrm{s})}(0))$
 is defined by integration of
 Eqs.~(\ref{eqn:motion_tascl_q}) and (\ref{eqn:motion_tascl})
 with $C=1$ over the period $T$.
 The initial state has to satisfy the condition
\begin{eqnarray}
 \mathcal{F}_{1,T}(\mathbf{q}^{(\mathrm{s})}(0),\mathbf{p}^{(\mathrm{s})}(0))
 =(\mathbf{q}^{(\mathrm{s})}(0),\mathbf{p}^{(\mathrm{s})}(0)).
\label{eqn:map_for_stationary_db}
\end{eqnarray}
 This is an equation for
 $(\mathbf{q}^{(\mathrm{s})}(0),\mathbf{p}^{(\mathrm{s})}(0))$.
 Solve Eq.~(\ref{eqn:map_for_stationary_db})
 by the Newton method with using $(a_1,\cdots,a_{N},0,\cdots,0)$
 as an initial guess for $(\mathbf{q}^{(\mathrm{s})}(0),\mathbf{p}^{(\mathrm{s})}(0))$. 
\label{enm:step3}

\item Construct an approximate traveling DB with velocity
 $v_\mathrm{DB} = r/s$ [site/period]
 in the symmetric lattice ($C=1$) under the constraint $U_{N/2}=0$
 by adding small perturbation to the stationary DB
 obtained in step (\ref{enm:step3}).
 The parameters $r$ and $s$ are integers.
 This means that the traveling DB propagates $r$ lattice spacings
 during $s$ internal oscillating periods $sT$, where $T = 2\pi/\omega_\mathrm{DB}$.
 It is natural to take the perturbation parallel to
 the direction $d\mathcal{T}_\lambda[\mathbf{U}]/d\lambda$, 
 since the map $\mathcal{T}_\lambda$ 
 represents a translational shift of DB along the lattice.
 Let $\delta U_m$ be each component of the perturbation vector,
 and we set
\begin{eqnarray}
 \delta U_m = -i m U_m\cdot \delta l,\quad m=-N/2+1,\cdots,N/2-1
\end{eqnarray}
 from a simple calculation of $d/d\lambda(U_m \exp(-im\lambda)) = -im U_m\exp(-im\lambda)$
 and $\delta U_{N/2}=0$ from $U_{N/2}=0$.
 The parameter $\delta l$ determines the magnitude of perturbation.
 The perturbation $\delta \mathbf{p}=(\delta p_1, \delta p_2,\cdots, \delta p_N)$
 in the physical space is given by
\begin{eqnarray}
 \delta p_n = -\delta l\frac{(-1)^n}{\sqrt{N}}
 \sum_{m=-N/2+1}^{N/2}{imU_m \exp\left[{-i\frac{2\pi n}{N}m}\right]},
 \quad n=1,2,\cdots, N.\nonumber\\
\label{eqn:variable_transformation_to_dp}
\end{eqnarray}
Fig.~\ref{fig:vdb_in_perturbed_db} shows the relation between the parameter $\delta l$ and the velocity $v_\mathrm{DB}$ of the traveling DB constructed by the perturbation (\ref{eqn:variable_transformation_to_dp}). 
The detailed procedure for estimating the velocity $v_\mathrm{DB}$ of the approximate traveling DB is described in \ref{sec:estimate_vdb}.
 It is found that the velocity $v_\mathrm{DB}$ is proportional to $\delta l$ in a certain range.
  Therefore, the parameter $\delta l$ is adjusted so that the traveling DB has
 the prescribed velocity $v_\mathrm{DB}$ (cf.~\ref{sec:estimate_vdb}).
 We obtain the initial state of the approximate traveling DB solution as
 $\bar{\mathbf{X}}^{(\mathrm{t})}(0)=(\mathbf{q}^{(\mathrm{s})}(0)
 ,\mathbf{p}^{(\mathrm{s})}(0)+\delta\mathbf{p})$.
\label{enm:step4}

\item
 Consider the TASCL with $C\in[0,1]$.
 For the prescribed values $T_\mathrm{DB}=2\pi/\omega_\mathrm{DB}$ and $v_\mathrm{DB}=r/s$,
 define the map $\mathcal{M}_C:\mathbb{R}^{2N}\to\mathbb{R}^{2N}$
 as follows:
\begin{eqnarray}
 \mathcal{M}_C = (-1)^r \mathcal{S}_r\circ \mathcal{F}_{C,sT},
\end{eqnarray}
 where $\mathcal{S}_r:\mathbb{R}^{2N}\to \mathbb{R}^{2N}$
 is the map that represents the cyclic permutation as follows:
\begin{eqnarray}
 \mathcal{S}_r(q_1,\cdots,q_N,p_1,\cdots,p_N)
 =(q_{1-r},\cdots,q_{N-r},p_{1-r},\cdots,p_{N-r}),
 \label{eqn:sr}
\end{eqnarray}
 Note that if the index $i-r\ (i=1,2,\cdots, N)$ of $q_{i-r}$ in RHS is not positive, it should be interpreted as $i-r+N$ since we consider periodic boundary conditions. 
 Let $(\mathbf{q}^{(\mathrm{t})}(0),\mathbf{p}^{(\mathrm{t})}(0))$ be
 the initial state of the traveling DB
 with $\omega_\mathrm{DB}$ and $v_\mathrm{DB}$. 
 The initial state has to satisfy the condition
\begin{equation}
 \mathcal{M}_C(\mathbf{q}^{(\mathrm{t})}(0),\mathbf{p}^{\mathrm{(t})}(0))
 =(\mathbf{q}^{(\mathrm{t})}(0),\mathbf{p}^{\mathrm{(t})}(0)).
\label{eqn:equation_for_moving_db}
\end{equation}
 This is an equation for
 $(\mathbf{q}^{(\mathrm{t})}(0),\mathbf{p}^{(\mathrm{t})}(0))$.
 It is possible to solve it by using the Newton method
 to find the solution
 precisely.
 Denote the solution of Eq.~(\ref{eqn:equation_for_moving_db}) with
 $\mathbf{X}_C^{(\mathrm{t})}(0)=(\mathbf{q}^{(\mathrm{t})}(0),\mathbf{p}^{(\mathrm{t})}(0))$.
\label{enm:step5}

\item 
 Consider the symmetric lattice ($C=1$).
 Solve Eq.~(\ref{eqn:equation_for_moving_db}) 
 by the Newton method
 with using $\bar{\mathbf{X}}^{(\mathrm{t})}(0)$ in step (\ref{enm:step4})
 as an initial guess
 to obtain $\mathbf{X}_1^{(\mathrm{t})}(0)$. 
\label{enm:step6}

\item Continue the solution $\mathbf{X}_1^{(\mathrm{t})}(0)$ in step (\ref{enm:step6})
 to the solution $\mathbf{X}_0^{(\mathrm{t})}(0)$
 for the FPU-$\beta$ lattice ($C=0$).
 This continuation is performed
 by repeatedly solving Eq.~(\ref{eqn:equation_for_moving_db})
 with gradually reducing the parameter $C$ until $C=0$. 
 Let $\Delta C>0$ be a small constant and
 $\mathbf{X}_C^{(\mathrm{t})}(0)$ be the solution of 
 Eq.~(\ref{eqn:equation_for_moving_db}) for $C$.
 Equation~(\ref{eqn:equation_for_moving_db}) for $C-\Delta C$
 can be solved by using $\mathbf{X}_C^{(\mathrm{t})}(0)$ as an initial guess
 for the Newton method.
 \label{enm:step7}
\end{enumerate}
%

   In steps (\ref{enm:step3}), (\ref{enm:step6}), and (\ref{enm:step7}),
 we have to find a fixed point $\mathbf{z} = (\mathbf{q}(0),\mathbf{p}(0))$
 of a map by solving the equation of the form
\begin{eqnarray}
 \mathcal{F}[\mathbf{z}] = \mathbf{z},
\label{eqn:equation_map_basic}
\end{eqnarray}
 where
 $\mathcal{F}:\mathbb{R}^{2N}\to\mathbb{R}^{2N}$
 is a continuously differentiable map.
 A fixed point $\mathbf{z}$ can be found by using the Newton method
 described below.
 Let $\mathbf{z}_0$ be a point that is close to
 the fixed point $\mathbf{z}$ of the map $\mathcal{F}$.
 Let $\mathbf{\Delta} = \mathbf{z}-\mathbf{z}_0$ be the deviation.
 Substituting $\mathbf{z} = \mathbf{z}_0 + \mathbf{\Delta}$
 into Eq.~(\ref{eqn:equation_map_basic})
 and performing the Taylor expansion with respect to $\mathbf{\Delta}$,
 we obtain
\begin{eqnarray}
 \mathcal{F}[\mathbf{z}_0]
 + D\mathcal{F}\cdot \mathbf{\Delta} - (\mathbf{z}_0 + \mathbf{\Delta}) \simeq \mathbf{0},
\label{eqn:equation_map_basic2}
\end{eqnarray}
 where $D\mathcal{F}$ is the Jacobian matrix of $\mathcal{F}$ evaluated at $\mathbf{z}_0$.
 From (\ref{eqn:equation_map_basic2}), we obtain
\begin{eqnarray}
 \mathbf{\Delta} = -(D\mathcal{F}-I)^{-1}(\mathcal{F}[\mathbf{z}_0]-\mathbf{z}_0).
\label{eqn:equation_correction_vector}
\end{eqnarray}
 Equation (\ref{eqn:equation_correction_vector}) gives
 the improved approximation $\mathbf{z}_0^{'} = \mathbf{z}_0 +\mathbf{\Delta}$.
 We can obtained an accurate numerical solution
 of Eq.~(\ref{eqn:equation_map_basic})
 by repeating this calculation until $|\mathbf{\Delta}|$ becomes sufficiently small.

 In the case of temporal evolution map $\mathcal{F}_{C,T}$,
 its Jacobian matrix $D\mathcal{F}_{C,T}$
 can be evaluated from the variational equation of
 Eqs.~(\ref{eqn:motion_tascl_q}) and (\ref{eqn:motion_tascl}). 
 which is given by
\begin{eqnarray}
 \dot{\xi}_n &=& \eta_n,\\
 \dot{\eta}_n &=& \xi_{n+1}+\xi_{n-1}-2\xi_n\nonumber\\
 &+& 3b_1\left[(q_{n+1}-q_{n})^2(\xi_{n+1}-\xi_{n})+(q_{n-1}-q_n)^2(\xi_{n-1}-\xi_n)\right]
\nonumber
\\
 &+&3C\sum_{r=2}^{N/2}{b_r\left[(q_{n+r}-q_n)^2(\xi_{n+r}-\xi_n)+(q_{n-r}-q_n)^2(\xi_{n-r}-\xi_n)\right]} 
\label{eqn:variational_tascl}
\end{eqnarray}
 where $n=1,2,\cdots,N$,
 $\xi_n$ and $\eta_n$ are variations in $q_n$ and $p_n$, respectively.
 Integration of Eq.~(\ref{eqn:variational_tascl})
 over the period $T$
 induces the temporal evolution map
 $\mathcal{G}_T:\mathbb{R}^{2N}\to\mathbb{R}^{2N}$ given by
\begin{eqnarray}
 \mathcal{G}_T[(\mathbf{\xi}(0), \mathbf{\eta}(0))]
 = (\mathbf{\xi}(T), \mathbf{\eta}(T)),
\end{eqnarray} 
 The Jacobian matrix $D\mathcal{F}_{C,T}$ is given by
\begin{eqnarray}
 D\mathcal{F}_{C,T}
 = (\mathcal{G}_T[\mathbf{\Delta}_1]^t,
 \mathcal{G}_T[\mathbf{\Delta}_2]^t,\cdots,
 \mathcal{G}_T[\mathbf{\Delta}_{2N}]^t),
\end{eqnarray}
 where the superscript $t$ stands for the transposition
 and $\mathbf{\Delta}_n$ is $2N$ vector
 in which only $n$th component is one and the other components are zero.

  In the case of map $\mathcal{M}_C$,
 its Jacobian matrix $D\mathcal{M}_C$,
 which is needed in steps (\ref{enm:step6}) and (\ref{enm:step7}), is given by
\begin{eqnarray}
 D\mathcal{M}_C =  (-1)^{r}S_r\circ (\mathcal{G}_{sT}[\mathbf{\Delta}_1]^t,
 \mathcal{G}_{sT}[\mathbf{\Delta}_2]^t,\cdots,
 \mathcal{G}_{sT}[\mathbf{\Delta}_{2N}]^t).
\end{eqnarray}
where $S_r$ is the matrix that represents the map (\ref{eqn:sr}).

 It should be noted that Eq.~(\ref{eqn:map_for_stationary_db}) in step (\ref{enm:step3}) is degenerate because of the translational invariance of equation of motion due to the conservation of total angular momentum and the arbitrariness of spatial symmetry of profile of stationary DB due to the extra conserved quantity of the symmetric lattices (see Eq.~(\ref{eqn:extra_conserved_quantity})). Therefore, we perform the Newton method under the constraint of $\sum_{n=1}^{N}q_n =0$ and keeping the spatial symmetry of profile of stationary DB, i.e., the even or odd mode.
Moreover, Eq.~(\ref{eqn:map_for_stationary_db}) is degenerate because of the arbitrariness of the initial point along the periodic orbit. In order to remove this degeneracy, we consider the constraint of $\mathbf{p}^{(\mathrm{s})}(0)=\mathbf{0}$ \cite{flach2008a}.
 
 As the same manner, Eq.(\ref{eqn:equation_for_moving_db}) in steps (\ref{enm:step6}) and (\ref{enm:step7}) is degenerate because of the translational invariance of equation of motion and the arbitrariness of the initial point along the trajectory of the traveling DB. In order to remove this degeneracy, we perform the calculation under the constraints of $\sum_{n=1}^{N}q_n = \sum_{n=1}^{N}p_n =0$, and $q_l >0$ and $p_l =0$, where $l$ is the index of a particle that has the maximum amplitude.

   Examples of the numerical solutions obtained
 by the above-mentioned procedure are presented
 in Fig.~\ref{fig:energy_profile2}.
 The internal period of DB is $T_\mathrm{DB}=2$
 and the velocity $v_\mathrm{DB}=1/10$. 
 Fig.~\ref{fig:energy_profile2}	shows
 particle energy profiles of DBs
 with different values of $C$. 
 The particle energy is defined by
 \begin{eqnarray}
 	e_n &=& \frac{1}{2}{p_{n}^2} +
 \frac{1}{4}\left[(q_{n}-q_{n-1})^2+(q_{n+1}-q_n)^2\right]\nonumber\\
 &&+\frac{b_1}{8} \left[(q_{n}-q_{n-1})^4+ (q_{n+1}-q_n)^4\right]
\nonumber\\
 &&+\frac{C}{8}\sum_{r=2}^{N/2}{b_r \left[(q_{n}-q_{n-r})^4+(q_{n+r}-q_n)^4 \right]}.
 \end{eqnarray}
 In the symmetric case ($C=1$),
 the traveling DB has no constant tail.
 By decreasing $C$, the spatially extended tail gradually appears. 
 The trajectory of averaged center position 
 of traveling DB with different $C$ 
 is presented in Fig.~\ref{fig:trajectory}.
 The center position of traveling DB is defined by
 \begin{eqnarray}
 	x(t) = \sum_{n=1}^{N}{n e_n}.
 \end{eqnarray}
 We perform the short-time average of $x(t)$ by
 \begin{eqnarray}
     X(t)=\frac{1}{2\tau}\int_{t-\tau}^{t+\tau}{x(t)dt}	
 \end{eqnarray}
 in order to reduce fluctuations of $x(t)$ due to the internal vibration of traveling DB.
 Figure \ref{fig:trajectory} shows the averaged center position $X(t)$.
 The DB travels with a constant velocity in the symmetric case ($C=1$).
 The slope of the trajectory is $1/20$, which is equal to $v_\mathrm{DB}/T$.
 In the FPU-$\beta$ lattice case ($C=0$), the velocity of DB periodically varies,
 but the averaged slope of the trajectory still coincides with $1/20$.
 These numerical results in Figs.~\ref{fig:energy_profile2} and \ref{fig:trajectory}
 demonstrate that
 the proposed calculation method works well
 and successfully computes the traveling DB in the FPU-$\beta$ lattice.

 In step (\ref{enm:step4}), we can obtain a good approximate traveling DB in the PISL which has a constant velocity. This is quite different from the case in the FPU-$\beta$ lattice because a traveling DB constructed from the perturbation gradually decreases its velocity. An advantage of the proposed method is that it is possible to obtain the traveling DB with a constant velocity easily.
  
 In this section, we focus on the calculation method for the traveling DB in the FPU-$\beta$ lattice. The proposed calculation method may apply to compute traveling DBs by the continuation in the other nonlinear lattices such as the nonlinear Klein-Gordon lattices and FPU-$\alpha\beta$ lattice, provided that no bifurcation occurs during the continuation from the PISL.

\begin{figure}[htpb]
\centering
\includegraphics{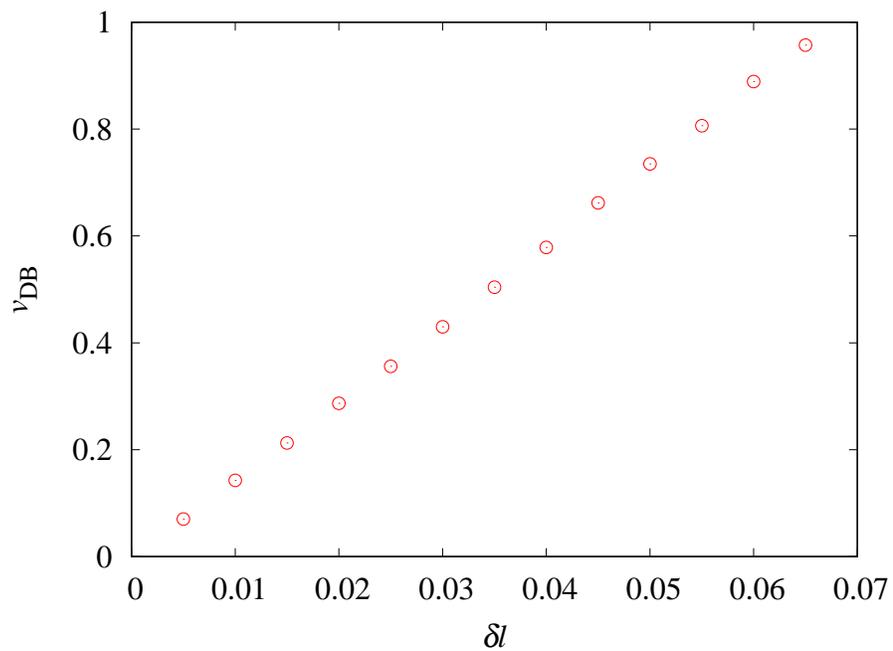}
\caption{Relation between $\delta l$ and the velocity of approximate traveling DB $v_\mathrm{DB}$. The period of internal vibration of the stationary DB is $T_\mathrm{DB}=2$.}
\label{fig:vdb_in_perturbed_db}
\end{figure}
\clearpage
\begin{figure}[htpb]
\centering
	\begin{tabular}{cc}
		\includegraphics[width=0.45\linewidth]{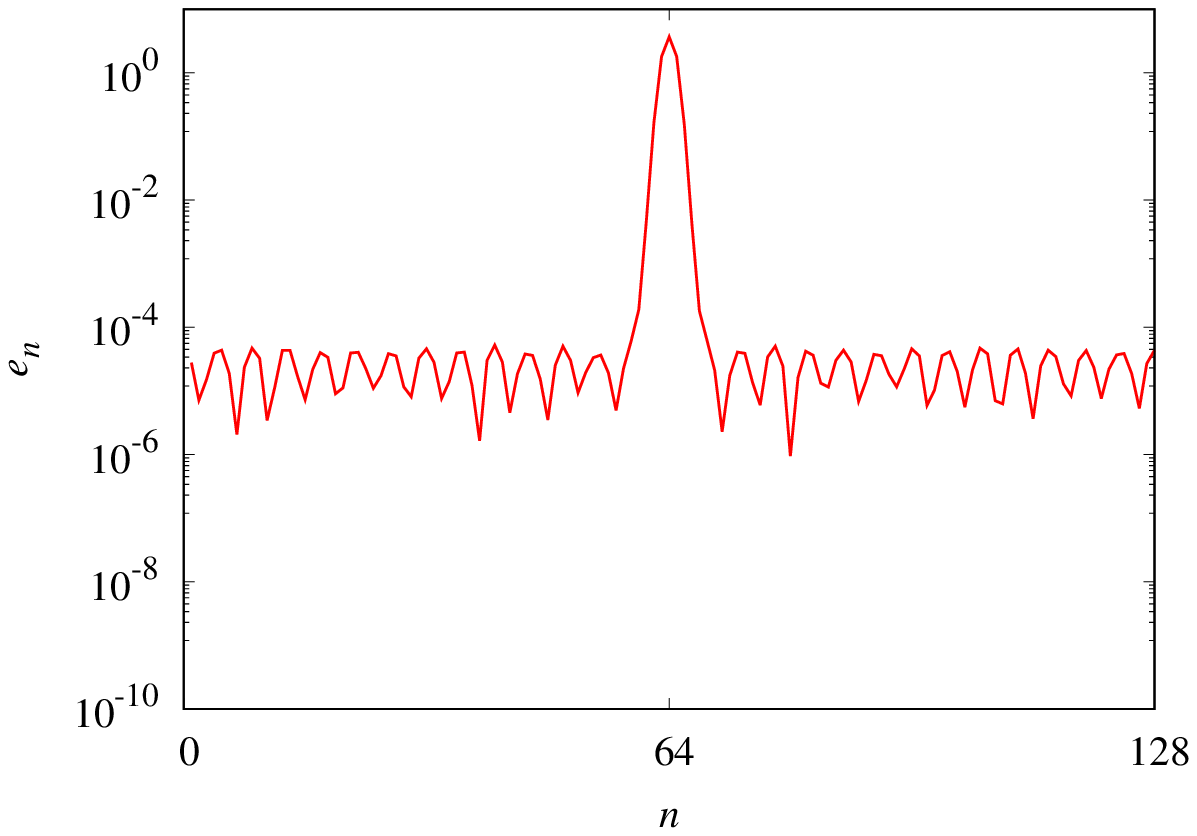}&
		\includegraphics[width=0.45\linewidth]{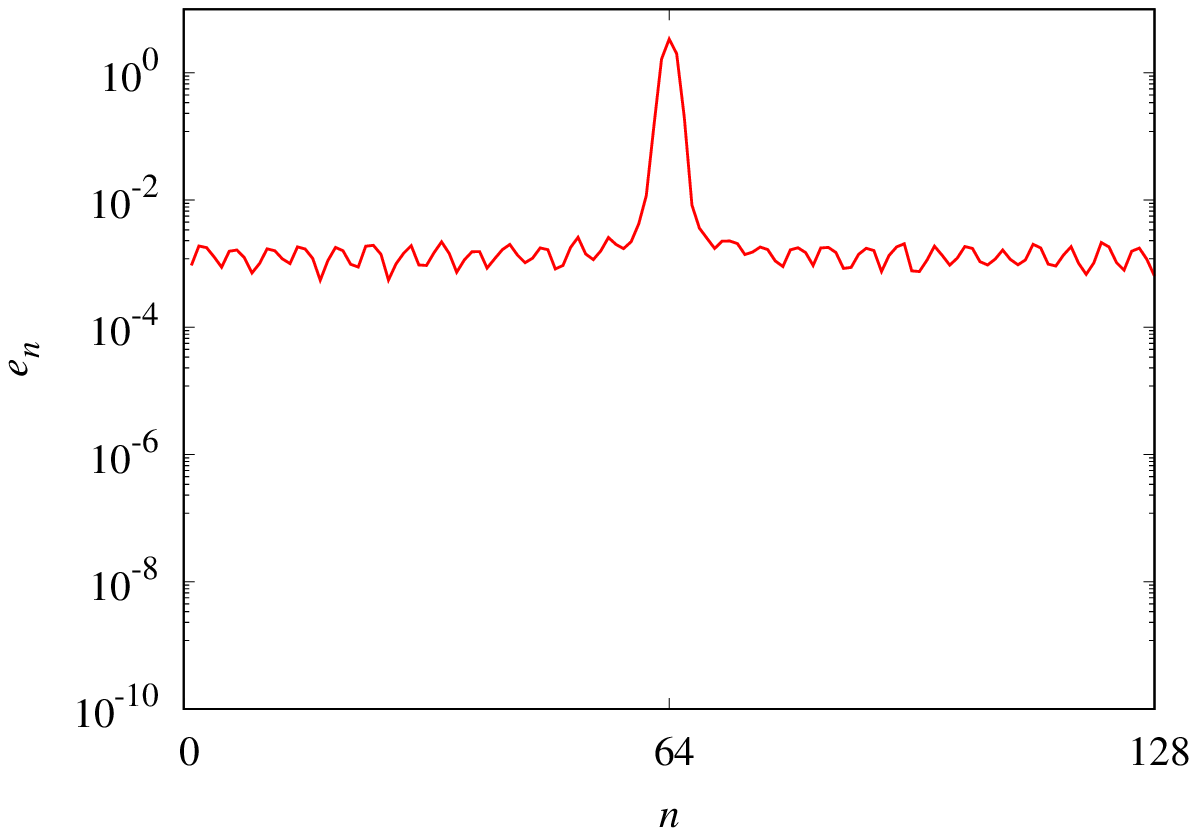}\\
		(a) $C=0.0$&(b) $C=0.2$\\
		\includegraphics[width=0.45\linewidth]{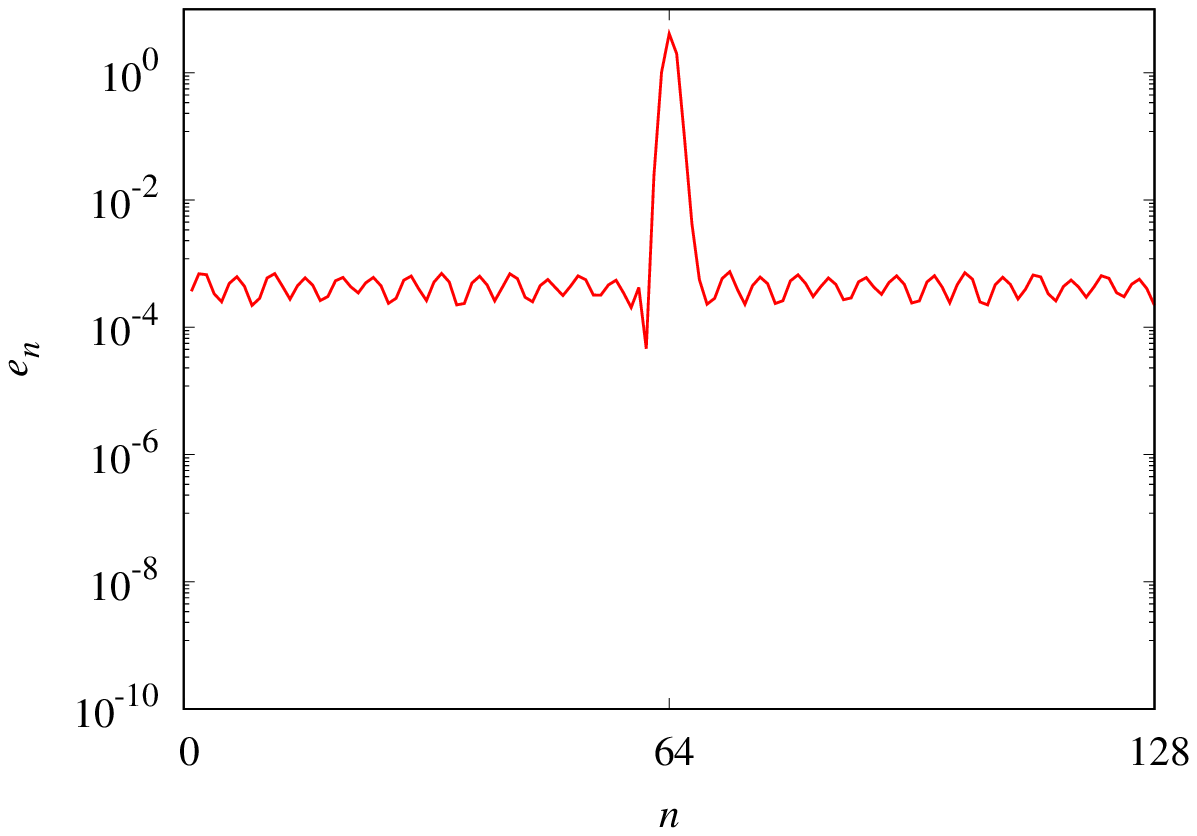}&
		\includegraphics[width=0.45\linewidth]{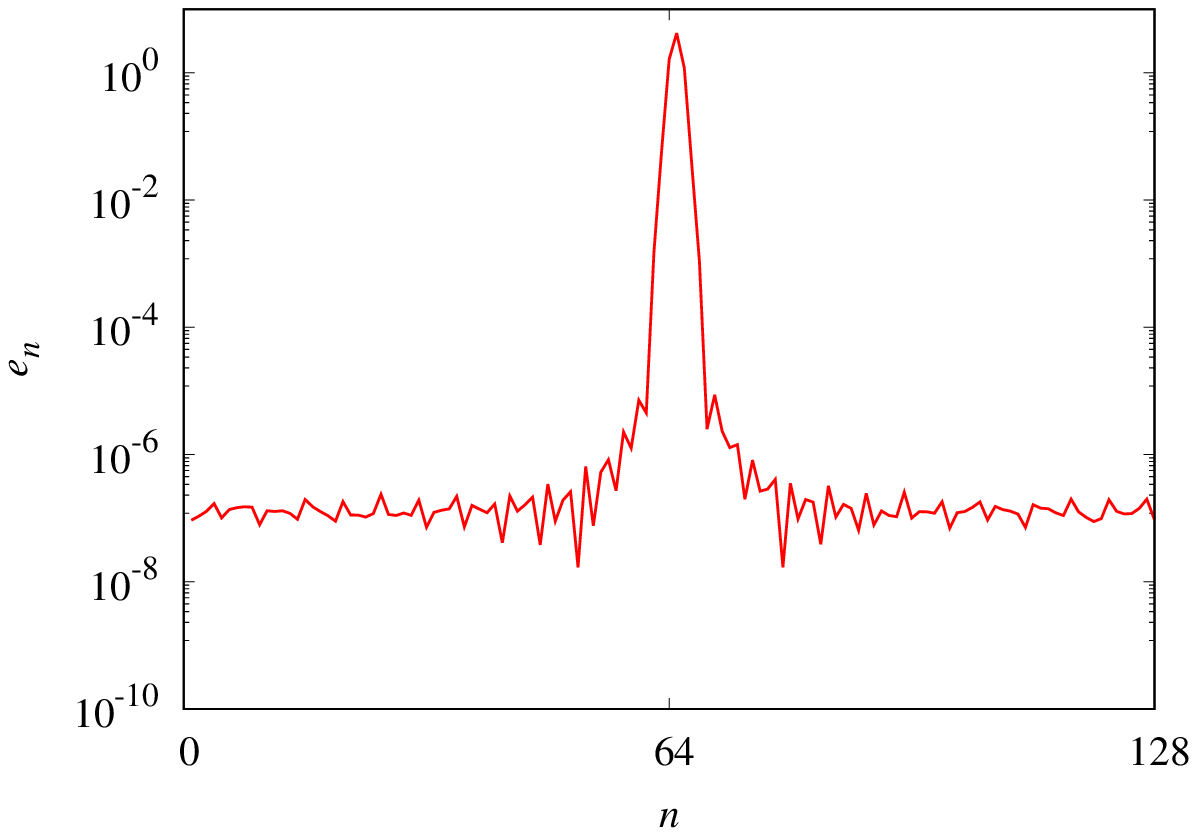}\\
		(c) $C=0.4$&(d) $C=0.6$\\
		\includegraphics[width=0.45\linewidth]{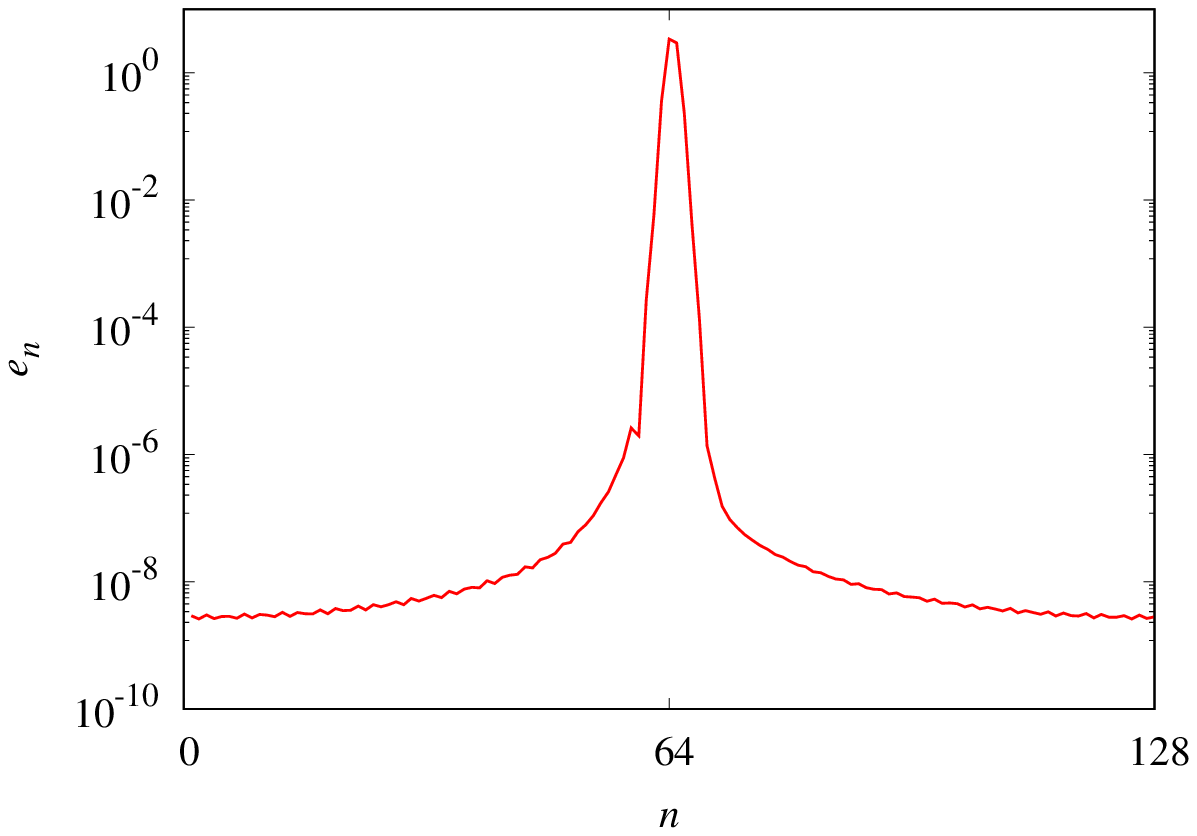}&
		\includegraphics[width=0.45\linewidth]{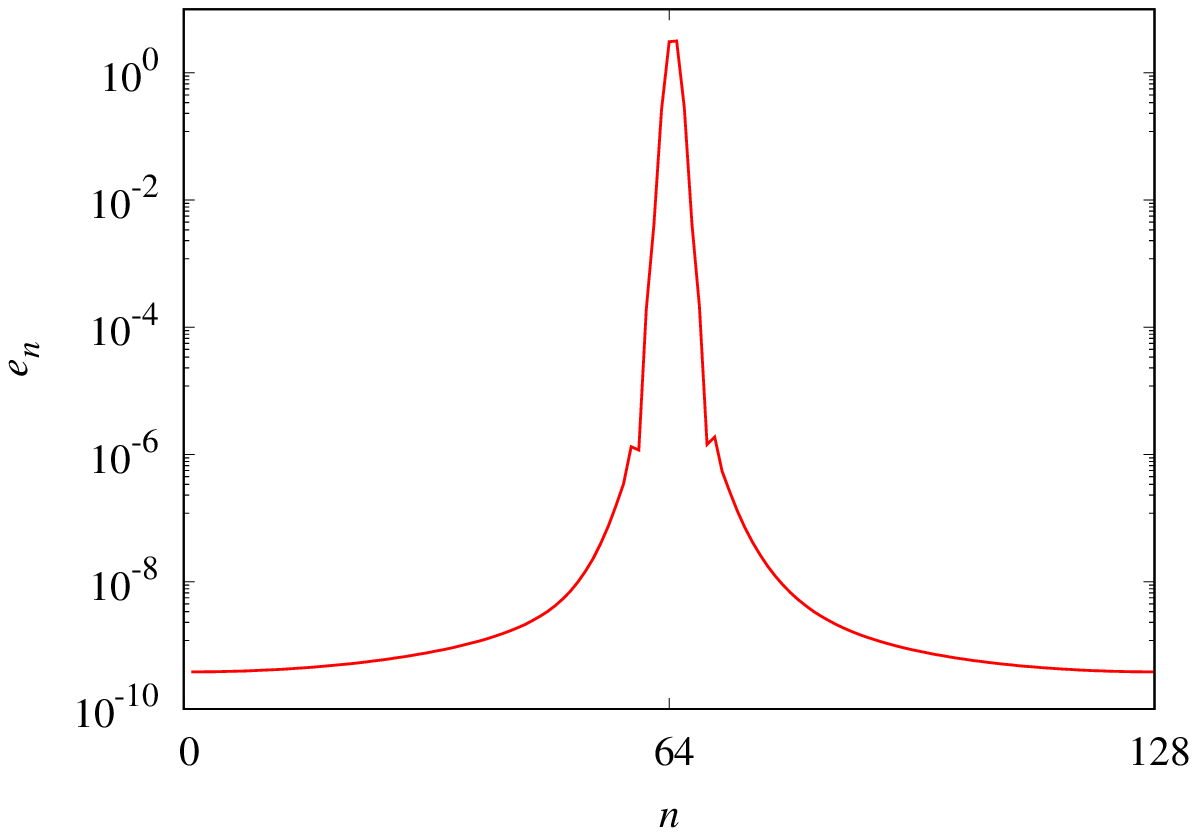}\\
		(e) $C=0.8$&(f) $C=1.0$\\
	\end{tabular}
	\caption{Energy profile of DBs obtained by the iteration method, $T_\mathrm{DB}=2$ and $v_\mathrm{DB}=1/10$.}
\label{fig:energy_profile2}
\end{figure}

\begin{figure}
    \centering
\begin{tabular}{c}
	\includegraphics{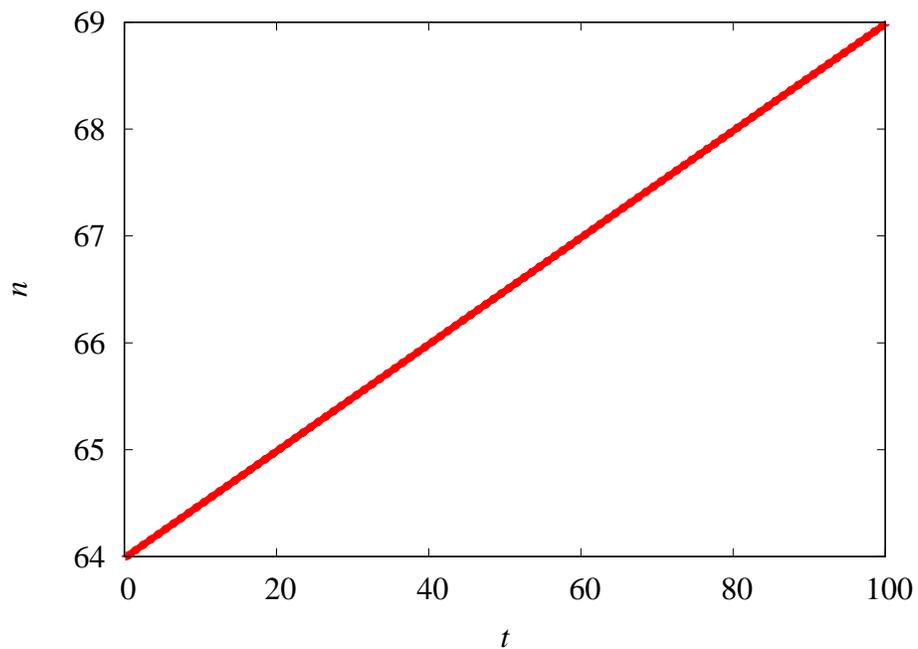}\\
	(a)$C=1$\\
	\includegraphics{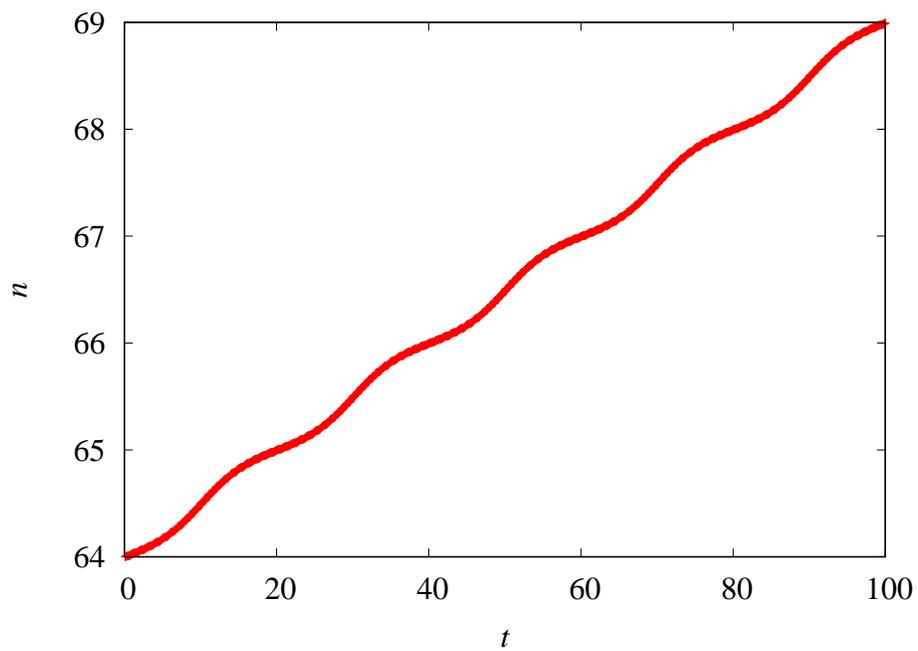}\\
	(b)$C=0$
\end{tabular}
	\caption{Trajectory of traveling DB with $v_\mathrm{DB}=1/10$ and $T_\mathrm{DB}=2\pi/\omega_\mathrm{DB}=2$ in (a)symmetric lattice ($C=1$) and (b)FPU-$\beta$ lattice ($C=0$) with $N=128$.}
	\label{fig:trajectory}
\end{figure}
\clearpage

%
%
\section{Truncated PISL and effects of breaking symmetry}\label{sec:truncated}
It is useful for investigating effects of breaking the symmetry to construct {\it a truncated PISL}, in which only up to $M$-th nearest neighbor interactions are considered:
\begin{eqnarray}
 H &=& \frac{1}{2}\sum_{n=1}^{N}{p_{n}^2} +\frac{1}{2}\sum_{n=1}^{N}{\left[\mu_0q_n^2+\mu_1(q_{n+1}-q_n)^2\right]}\nonumber\\&+&
  \frac{1}{4}\sum_{n=1}^{N}{\sum_{r=1}^{M}{b_r(q_{n+r}-q_n)^4}}.
 \label{eqn:explicit_approximated_symmetric_lattice}
\end{eqnarray}

The lattice (\ref{eqn:explicit_approximated_symmetric_lattice}) can be regarded as the PISL with the perturbation term $\Delta H$ as follows: 
\begin{eqnarray}
	\Delta H(M) = -\frac{1}{4} \sum_{n=1}^{N}{\sum_{r=M+1}^{N/2}{b_1 (q_{n+r}-q_n)^4}}.
	\label{eqn:perturbation_H}
\end{eqnarray}
This perturbation breaks the symmetry of lattice. The parameter $M$ corresponds to the magnitude of perturbation. As the parameter $M$ becomes smaller, the magnitude of perturbation becomes larger. 

When 
the perturbation (\ref{eqn:perturbation_H}) is introduced, the averaged center position $X(t)$ of a traveling DB deviates from the straight line $x_\mathrm{s}(t)=(v_\mathrm{DB}/T_\mathrm{DB})t$ which corresponds to  the case of a constant velocity.
Figure \ref{fig:trajectory_truncated} shows 
the deviation for different values of $M$.
 The magnitude of deviation becomes larger as $M$ decreases, i.e., the magnitude of perturbation becomes larger. 
It can be concluded that one of the symmetry breaking effects is the variation of DB's velocity.
In addition to this effect, the symmetry breaking causes the velocity loss of approximate traveling DB (cf.~Supplemental Material of \cite{doi2016}) and a 
    degradation of the ballistic thermal transport observed in the symmetric case \cite{yoshimura2019}. 
\begin{figure}[ht]
\centering
\includegraphics{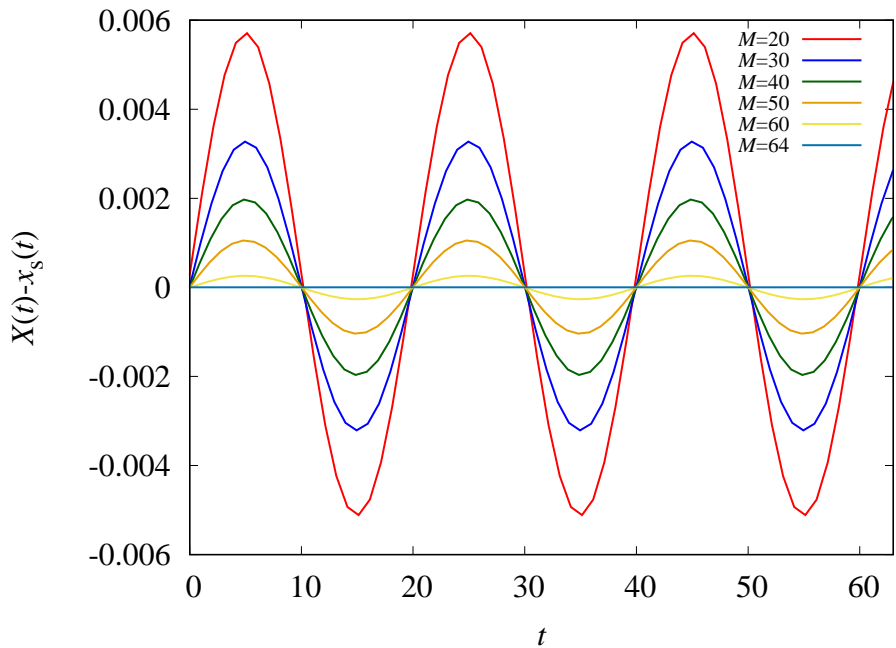}
\caption{Deviation of the averaged position of traveling DB $X(t)$ from the straight trajectory $x_\mathrm{s}(t)=(v_\mathrm{DB}/T_\mathrm{DB})t$ with $v_\mathrm{DB}=1/10$ and $T_\mathrm{DB}=2\pi/\omega_\mathrm{DB}=2$ in the truncated PISL with $N=128$, and $M=20, 30, 40, 50, 60$. The case of full PISL $M=64$ is also presented.}
\label{fig:trajectory_truncated}
\end{figure}

%
%
\section{Decomposition of Eq.~(\ref{eqn:equation_for_s})}\label{sec:preliminary}
We have assumed that $N$ is a multiple of 4.
Under this assumption, equations (\ref{eqn:equation_for_s}) can be rewritten to
\begin{eqnarray}
   c_{i}c_{j}c_{k}c_{l}b_1-s_{2i}s_{2j}s_{2k}s_{2l}b_2\ldots-s_{\frac{N}{2}i}s_{\frac{N}{2}j}s_{\frac{N}{2}k}s_{\frac{N}{2}l}b_{N/2}
    = 0.
\end{eqnarray}
The function $\psi^{(i,j,k,l)}(\mathbf{b})$ can be decomposed into two parts as follows:
\begin{equation}
  \psi^{(i,j,k,l)}(\mathbf{b})=\psi_\mathrm{odd}^{(i,j,k,l)}(\mathbf{b}_\mathrm{odd})-\psi_\mathrm{even}^{(i,j,k,l)}(\mathbf{b}_\mathrm{even}),
\end{equation}
where
\begin{eqnarray}
  \psi_\mathrm{odd}^{(i,j,k,l)}(\mathbf{b}_\mathrm{odd})&=&\sum_{s=1}^{N/4}b_{2s-1}c_{(2s-1)i}c_{(2s-1)j}c_{(2s-1)k}c_{(2s-1)l},\\
  \psi_\mathrm{even}^{(i,j,k,l)}(\mathbf{b}_\mathrm{even})&=&\sum_{s=1}^{N/4}b_{2s}s_{2si}s_{2sj}s_{2sk}s_{2sl},
\end{eqnarray}
and, $\mathbf{b}_\mathrm{odd} = [b_1, b_3, \cdots, b_{N/2-1}]^{T}$ and $\mathbf{b}_\mathrm{even} = [b_2, b_4, \cdots, b_{N/2}]^{T}$ are $N/4$ vectors.

We will discuss the equations for $b_{2s-1}$ and $b_{2s}$ separately in the following sections.

\section{Equations for $b_{2s-1}$}\label{sec:odd_equation}
The set $S$ can be
divided into two subsets $T_1 = \{(i=0,j,k,l)|-N_h \le j \le k \le l \le
N_h, j+k+l=N\} \subset S$ and $T_2 = S\setminus T_1$. 
When $(i,j,k,l)\in T_1$, all terms 
corresponding to $s_0$ vanish. Therefore, $\psi^{(i,j,k,l)}(\mathbf{b})=0\ \forall (i,j,k,l)\in T_1$ reduces to
\begin{eqnarray}
 \psi^{(i,j,k,l)}_\mathrm{odd}(\mathbf{b}_\mathrm{odd}) &=&
  \sum_{s=1}^{N/4}b_{2s-1}c_{(2s-1)i}c_{(2s-1)j}c_{(2s-1)k}c_{(2s-1)l}=0. \quad\forall (i,j,k,l)\in T_1.\nonumber\\
  \label{eqn:equation_for_symmetric_s1}
\end{eqnarray}

We consider the equations $\psi^{(i,j,k,l)}(\mathbf{b})=0\ \forall (i,j,k,l)\in S_1=\{(0,n+1,N/2-n, N/2-1)|1\le n \le N/4-1\}\subset T_1$. 
Substituting $(i,j,k,l)=(0,n+1,N/2-n,N/2-1)$ into the LHS of (\ref{eqn:equation_for_symmetric_s1}), we obtain
\begin{eqnarray}
	\psi_\mathrm{odd}^{(0,n+1,N/2-n,N/2-1)}(\mathbf{b}_\mathrm{odd}) &=&
	\sum_{s=1}^{N/4}b_{2s-1}c_0c_{(2s-1)(n+1)}c_{(2s-1)(N/2-n)}c_{(2s-1)(N/2-1)}\nonumber\\
	&=&\frac{1}{4}\sum_{s=1}^{N/4}{b_{2s-1}\left[
				      -1-c_{2(2s-1)(n+1)}+c_{2(2s-1)}+c_{2(2s-1)n}\right]}\nonumber\\
\end{eqnarray}
Therefore, equation (\ref{eqn:equation_for_symmetric_s1}) reduces to
\begin{eqnarray}
\frac{1}{4}\sum_{s=1}^{N/4}{b_{2s-1}\left[
				      -1-c_{2(2s-1)(n+1)}+c_{2(2s-1)}+c_{2(2s-1)n}\right]}=0.\nonumber
 \\(n=1,2,\ldots,\frac{N}{4}-1)
\label{eqn:coefficient_asymmetric_term_i0}
\end{eqnarray}
Eq.(\ref{eqn:coefficient_asymmetric_term_i0}) can be rewritten into the matrix form,
\begin{eqnarray}
 \frac{1}{4}M_1 A \mathbf{b}_\mathrm{odd}=0,
 \label{eqn:equation_for_symmetry_odd_matrix}
\end{eqnarray}
where $A$ is the $N/4\times N/4$ matrix whose element is given by $A_{pq}=
c_{2(p-1)(2q-1)}$. $M_1$ is the $(N/4-1) \times N/4$ matrix defined by
\begin{eqnarray}
 M_1 = \left[
	\begin{array}{ccccccccc}
	 -1&2&-1 &0 & & &&\cdots &0 \\
	 -1&1&1 &-1 &0 & & &\cdots &0 \\
	 -1&1&0 & 1 &-1&0& &\cdots&0 \\
	 \vdots&\vdots&&&&&&\ddots& \\
	 &&&&&&&& \\
	 -1&1&0 & & &\cdots&0&1 &-1 \\
	 -1&1&0 & & &\cdots &0 &0 &1 \\
	\end{array}
       \right].
\end{eqnarray}
Note that in the last row in $M_1$, we use the relation $c_{2(2s-1)N/4}=0$.

For the following discussion
we introduce the $(N/4-1)\times (N/4-1)$ matrix $P_1$ defined by:
\begin{eqnarray}
 P_1 &=& \frac{1}{N/4}\left[
	\begin{array}{ccccc}
	 1&1&&\cdots&1\\
	 -(N/4\mathrel{-}2)&2&&\cdots&2\\
	 -(N/4\mathrel{-}3)&-(N/4\mathrel{-}3)&3&\cdots&3\\
	 \vdots&\vdots&&\ddots&\vdots\\
	 -1&-1&\cdots&-1&N/4\mathrel{-}1\\
	\end{array}
		\right].
\end{eqnarray}

\begin{lemma}
$P_1$ is regular matrix.
\label{lemma:regularity_p}
\end{lemma}
\begin{proof}
Let a matrix $T_{1,p,q}$ be the elementary matrix which represents the elementary row operation of adding $q$ times of the first
row to the $p$-th row.
Applying $T_{1,m,N/4-m}\quad (m=2, 3, \cdots, N/4-1)$ to $P_1$ from left, we obtain an upper triangular matrix as
\begin{eqnarray}
	T_{1,N/4-1,1}T_{1,N/4-2,2}\cdots T_{1,2,N/4-2}P_1 = 
	\frac{1}{N/4}\left[
	\begin{array}{ccccc}
	 1&1&1&\cdots&1\\
	 0&N/4&N/4&\cdots&N/4\\
	 0&0&N/4&\cdots&N/4\\
	 \vdots&\vdots&&\ddots&\vdots\\
	 0&0&\cdots&0&N/4\\
	\end{array}
	      \right].\nonumber\\
\end{eqnarray}
The transformed upper triangular matrix is regular. Therefore $P_1$ is
regular matrix.
\end{proof}

Applying $P_1$ to $M_1$, we obtain the $(N/4-1) \times N/4$ matrix
\begin{eqnarray}
 P_1 M_1= \left[
	   \begin{array}{ccccc}
	   f_{N/4}(1) &1 &0 &\cdots &0 \\
	   f_{N/4}(2) &0 &1 &\cdots &0 \\
	   \vdots &\vdots & &\ddots &\vdots \\
	   f_{N/4}(N/4\mathrel{-}1)&0 &0 &\cdots &1 \\
	   \end{array}
		   \right]
		   \label{eqn:p1m1}
\end{eqnarray}
where $f_m(n) = -1 + n/m$. 
It is clear from Eq.(\ref{eqn:p1m1}) that the rank of $P_1 M$ is $N/4-1$.
Then, the following lemma holds:

\begin{lemma}
Rank of $M_1$ is $N/4-1$.
\label{lemma:rank_M1}
\end{lemma}
The following proposition for the equations $\psi_\mathrm{odd}^{(i,j,k,l)}(\mathbf{b}_\mathrm{odd})=0, \forall (i,j,k,l)\in S_1$ holds:
\begin{prop}
The equations $\psi^{(i,j,k,l)}_\mathrm{odd}(\mathbf{b}_\mathrm{odd})=0, \forall (i,j,k,l)\in S_1$ are the $N/4-1$ linearly independent equations for
$b_{2s-1} (s=1,2,\ldots,N/4)$.
\label{prop:independent_odd}
\end{prop}

\begin{proof}
The rank of matrix $A$ is $N/4$ (see \ref{sec:regularity_C}) and that of matrix $M_1$
is $N/4-1$ from Lemma \ref{lemma:rank_M1}. Therefore, the rank of matrix $M_1 A$ is $N/4-1$. Therefore, the equations $\psi^{(i,j,k,l)}_\mathrm{odd}(\mathbf{b}_\mathrm{odd})=0, \forall (i,j,k,l)\in S_1$ are $N/4-1$ linearly independent equations for $b_{2s-1} (s=1,2,\cdots,N/4)$.
\end{proof}

The equations $\psi^{(i,j,k,l)}_\mathrm{odd}(\mathbf{b}_\mathrm{odd})=0, \forall (i,j,k,l)\in S_1$ are $N/4-1$ linearly independent equations. On the other hand, the number of unknown
variables $b_{2s-1}$ is $N/4$. Therefore, we can express the nontrivial solutions $b_{2s-1}
(s=2,3,\ldots,N/4)$ in term of $b_1$ as follows:
\begin{eqnarray}
 b_{2s-1} = w_{2s-1}b_1,
 \label{eqn:solution_odd}
\end{eqnarray}
where $w_{2s-1}$ is a real constant. Substituting (\ref{eqn:solution_odd}) into (\ref{eqn:equation_for_symmetry_odd_matrix}) and then applying $P_1$ from the left side, we obtain the equations for $w_{2s-1}$:
\begin{eqnarray}
 \sum_{s=1}^{N/4}{\left[f_{N/4}(n) +c_{2n(2s-1)}\right]w_{2s-1}}=0
 \label{eqn:relation_odd}
\end{eqnarray}
for $n=1,2,\cdots,N/4-1$. Note that we define $w_1 = 1$.

It can be easily confirmed that Eq.(\ref{eqn:relation_odd}) also holds for the case of $n=0$ and $n=N/4$. 
We next check that Eq.(\ref{eqn:relation_odd}) also holds for a wider range of $n$. Define a new variable $n'=N/2-n$. Substituting
$n=N/2-n'$ into Eq.(\ref{eqn:relation_odd}), we obtain
\begin{eqnarray}
 \sum_{s=1}^{N/4}{\left[f_{N/4}(n') +c_{2n'(2s-1)}\right]w_{2s-1}}=0
 \label{eqn:relation_odd2}
\end{eqnarray}
for $n'=N/4+1, N/4+2, ..., N/2-1$.  This indicates that
Eq.(\ref{eqn:relation_odd}) also holds for $n=N/4+1,
N/4+2,\ldots,N/2-1$.  Finally, it is shown that Eq.(\ref{eqn:relation_odd}) holds for
the extended range $n=0,1,...\ldots,N/2-1$.

Next, we consider the equations $\psi^{(i,j,k,l)}_\mathrm{odd}(\mathbf{b}_\mathrm{odd})=0,\forall (i,j,k,l)\in T_1$.

\begin{prop}
The nontrivial solution $b_{2s-1} = w_{2s-1}b_1$ of the equations $\psi^{(i,j,k,l)}_\mathrm{odd}(\mathbf{b}_\mathrm{odd})=0,\forall (i,j,k,l)\in S_1$ also solves the other equations in $T_1$.
\label{prop:solve_on_t1}
\end{prop}
\begin{proof}
The set $T_1$ can be expressed by $(i,j,k,l)=(0,a+b,N/2-b,N/2-a)$, where
$a$ and $b$ are integer which are in the ranges of $2\le a \le N/6$ and $a \le b \le (N-2a)/4$.
Substituting solution (\ref{eqn:solution_odd}) and using
$\sum_{s=1}^{N/4}{c_{2n(2s-1)}w_{2s-1}} = -\sum_{s=1}^{N/4}{f_{N/4}(n)w_{2s-1}}$, 
the left hand side of equations $\psi^{(i,j,k,l)}_\mathrm{odd}(\mathbf{b}_\mathrm{odd})=0, \forall (i,j,k,l)\in T_1$ can be expressed as
\begin{eqnarray}
\psi&&^{(0,a+b,N/2-b,N/2-a)}_\mathrm{odd}(w_{2s-1}b_1)\nonumber
\\
&&=b_1\sum_{s=1}^{N/4}{w_{2s-1}\left[
 -1-c_{2(2s-1)(a+b)}+c_{2(2s-1)b}+c_{2(2s-1)a}\right]}\nonumber\\
&&=b_1\sum_{s=1}^{N/4}{w_{2s-1}\left[
 -1+f_{N/4}(a+b)-f_{N/4}(b)-f_{N/4}(a)\right]}\nonumber\\
&&=0
\label{eqn:coefficient_asymmetric_term_i0_3}
\end{eqnarray}
Therefore, the nontrivial solution $b_{2s-1}=w_{2s-1}b_1$ of the equations $\psi^{(i,j,k,l)}_\mathrm{odd}(\mathbf{b}_\mathrm{odd})=0, \forall (i,j,k,l)\in S_1$ also solves the other equations in $T_1$.	
\end{proof}

\section{Equations for $b_{2s}$}\label{sec:even_equation}
We consider $\psi^{(i,j,k,l)}(\mathbf{b})=0, \forall (i,j,k,l)\in S_2=\{(2-m,m,N/2-1,N/2-1)|m=1\ \mbox{or}\ 3\le m \le N/4+1\}$ under the
condition that $b_{2s-1} = w_{2s-1}b_1$ is the nontrivial solution of 
$\psi_{\mathrm{odd}}^{(i,j,k,l)}(\mathbf{b}_{\mathrm{odd}})=0, \forall (i,j,k,l)\in S_1$.
Substituting Eq.(\ref{eqn:solution_odd}) into
$\psi^{(i,j,k,l)}(\mathbf{b})=0,\quad\forall (i,j,k,l)\in S_2$, we obtain
\begin{eqnarray}
\psi_{\mathrm{even}}^{(i,j,k,l)}(\mathbf{b}_\mathrm{even})= \mathcal{R}^{(i,j,k,l)}(w_{1},w_{3},\cdots,w_{N/2-1}),\quad\forall (i,j,k,l)\in S_2,
\label{eqn:equation_for_symmetric_s2}
\end{eqnarray}
where $\mathbf{b}_\mathrm{even} = [b_2,b_4,\cdots,b_{N/2}]^{T}$ is the $N/4$-dimensional vector, and LHS and RHS of Eq.(\ref{eqn:equation_for_symmetric_s2}) are given as
\begin{eqnarray}
\psi_{\mathrm{even}}^{(i,j,k,l)}(\mathbf{b}_\mathrm{even})&=&\sum_{s=1}^{N/4}{b_{2s}s_{2si} s_{2sj} s_{2sk} s_{2sl}}, \label{eqn:s2_lhs}\\
 \mathcal{R}^{(i,j,k,l)}(w_{1},w_{3},\cdots,w_{N/2-1})&=&\psi_\mathrm{odd}^{(i,j,k,l)}(\{b_1 w_{2s-1}\}_s)\nonumber\\ &=&\sum_{s=1}^{N/4}b_1{w_{2s-1}
 c_{(2s-1)i} c_{(2s-1)j} c_{(2s-1)k} c_{(2s-1)l}}\label{eqn:s2_rhs},\nonumber\\
\end{eqnarray}
respectively.

 Substituting
$i=2-m$, $j=m$, $k=N/2-1$, $l=N/2-1$, Eq.~(\ref{eqn:s2_lhs})
can be rewritten as follows:
\begin{eqnarray}
\psi&& ^{(2-m,m,N/2-1,N/2-1)}_{\mathrm{even}}(\mathbf{b}_\mathrm{even})\nonumber\\
&&=
\sum_{s=1}^{N/4}{b_{2s}s_{2s(2-m)} s_{2sm} s_{2s(N/2-1)} s_{2s(N/2-1)}}\nonumber\\&&=\frac{1}{8}\sum_{s=1}^{N/4}\left[1-2c_{4s}+c_{8s}-c_{4s(m-2)}+2c_{4s(m-1)}-c_{4sm}\right]b_{2s}.
\end{eqnarray}
On the other hand, Eq.(\ref{eqn:s2_rhs}) becomes,
\begin{eqnarray}
 &&\mathcal{R}^{(2-m,m,N/2-1,N/2-1)}\nonumber\\
 &&=
 \sum_{s=1}^{N/4}b_1{w_{2s-1}
 c_{(2s-1)(2-m)} c_{(2s-1)m} c_{(2s-1)(N/2-1)} c_{(2s-1)(N/2-1)}}\nonumber\\
 &&=
 \frac{b_1}{8}\sum_{s=1}^{N/4}{w_{2s-1}[
 -1+2c_{2(2s-1)}-c_{4(2s-1)}-c_{(2s-1)(2m-4)}+2c_{(2s-1)(2m-2)}-c_{2m(2s-1)}]}.\nonumber\\
&&\quad\quad (m=1, \mbox{ or } 3\le m \le \frac{N}{4}+1)
 \label{eqn:equation_for_symmetric_s2_2}
\end{eqnarray}
Note that the sign of index of the fourth term in Eq.(\ref{eqn:equation_for_symmetric_s2_2}) can change depending on the value of $m$.
In the case that $m=1$, the index 
 becomes negative. The relation (\ref{eqn:relation_odd}) is only valid for nonnegative index of $c_n$. We use the relation $c_{-n}=c_{n}$ to avoid the negative index. Substituting (\ref{eqn:relation_odd}) into
(\ref{eqn:equation_for_symmetric_s2_2}), we obtain
\begin{eqnarray}
 \mathcal{R}^{(1,1,N/2-1,N/2-1)} = \frac{b_1}{N}\sum_{s=1}^{N/4}\omega_{2s-1}\equiv Wb_1.
 \label{eqn:definition_R}
\end{eqnarray}
When $m\ge3$, substituting (\ref{eqn:relation_odd}) into
(\ref{eqn:equation_for_symmetric_s2_2}), we obtain
\begin{eqnarray}
 &&\mathcal{R}^{(2-m,m,N/2-1,N/2-1)}\nonumber\\ &=&
 \frac{b_1}{8}\sum_{s=1}^{N/4}{w_{2s-1}[
 -1-2f_{N/4}(1)+f_{N/4}(2)+f_{N/4}(m-2)-2f_{N/4}(m-1)+f_{N/4}(m)]}\nonumber\\
 &=&0.
\end{eqnarray}

Therefore, Eq.(\ref{eqn:equation_for_symmetric_s2}) is rewritten as
\begin{eqnarray}
\left\{
\begin{array}{l}
 \displaystyle{\frac{1}{8}\sum_{s=1}^{N/4}{[3 -4c_{4s}+c_{8s}]}b_{2s} =
  Wb_{1}\quad(m=1)},\\
 \displaystyle{\frac{1}{8}\sum_{s=1}^{N/4}{[1-2c_{4s}+c_{8s}-c_{4s(m-2)}+2c_{4s(m-1)}-c_{4sm}]}b_{2s}=0.
  \quad(m=3,4,\ldots,\frac{N}{4}+1)}
\end{array}
\right.
  \label{eqn:equation_for_symmetric_s3}
\end{eqnarray}
The matrix form of Eq.~(\ref{eqn:equation_for_symmetric_s3}) is written as
follows:\begin{eqnarray}
 \frac{1}{8}\tilde{M}_2 \tilde{D} \mathbf{b}_\mathrm{even} = \mathbf{g},
 \label{eqn:equation_for_symmetry_even_matrix}
\end{eqnarray}
where $\tilde{M}_2$ is the $N/4 \times (N/4+1)$ matrix
\begin{eqnarray}
 \tilde{M}_2 = \left[
	\begin{array}{ccccccccc}
	 3&-4&1 &0 & & &&\cdots &0 \\
	 1&-3&3 &-1 &0 & & &\cdots &0 \\
	 1&-2&0 & 2 &-1&0& &\cdots&0 \\
	 1&-2&1 & -1 &2&-1&0&\cdots&0 \\
	 \vdots&\vdots&&&&&&\ddots& \\
	 &&&&&&&& \\
	 1&-2&1 &0&\cdots &-1&2&-1 &0 \\
	 1&-2&1 &0&\cdots &0&-1&2 &-1 \\
	 1&-2&1 &0& &\cdots&0&-2 &2 \\
	 \end{array}
       \right],
\end{eqnarray}
Note that in the last row in $\tilde{M}_2$, we use the relation $c_{4s(N/4+1)} = c_{4s(N/4-1)}$.
$\tilde{D}$ is the $(N/4+1) \times N/4$ matrix that element is given by
$D_{pq}=c_{4(p-1)q}$, and $\mathbf{g}$ is the $N/4$-dimensional vector given by
\begin{eqnarray}
 \mathbf{g} =\left[Wb_1, 0, \cdots, 0\right]^{T}.
\end{eqnarray}
Since $\tilde{M}_2$ and $\tilde{D}$ are not square matrix,
we convert these matrices to square ones in order to proceed the proof.

Let $\tilde{\mathbf{x}} = [\tilde{x}_1,\tilde{x}_2, \cdots, \tilde{x}_{N/4+1}]^T$ be the $(N/4+1)$ vector defined as follows:
\begin{eqnarray}
	\tilde{\mathbf{x}} = \tilde{D} \mathbf{b}_\mathrm{even},
	\label{eqn:bbar}
\end{eqnarray}
or
\begin{eqnarray}
	\tilde{x}_i = D_{i,1}b_2 + D_{i,2}b_4+\cdots+D_{i,N/4}b_{N/2}.\quad(i=1,2,\cdots,N/4+1).
	\label{eqn:bbar2}
\end{eqnarray}
Using the relation (\ref{eqn:bbar}), Eq.~(\ref{eqn:equation_for_symmetry_even_matrix}) is rewritten as
\begin{eqnarray}
	\frac{1}{8}\tilde{M}_2 \tilde{\mathbf{x}} = \mathbf{g}.
	\label{eqn:sub1}
\end{eqnarray}

Let $\mathbf{m}_i\quad(i=1,2,\cdots,N/4+1)$ be the column vectors of $\tilde{M}_2$. It can be easily checked that the last column vector $\mathbf{m}_{N/4+1} = -\sum_{j=1}^{N/4}{\mathbf{m}_j}$. Using this relation, Eq.~(\ref{eqn:sub1}) can be transformed into
\begin{equation}
	\frac{1}{8}M_2 \mathbf{x} = \mathbf{g},
	\label{eqn:sub2}
\end{equation}
where $M_2$ is the $N/4 \times N/4$ matrix
\begin{eqnarray}
 M_2 = \left[
	\begin{array}{ccccccccc}
	 3&-4&1 &0 & & &&\cdots &0 \\
	 1&-3&3 &-1 &0 & & &\cdots &0 \\
	 1&-2&0 & 2 &-1&0& &\cdots&0 \\
	 1&-2&1 & -1 &2&-1&0&\cdots&0 \\
	 \vdots&\vdots&&&&&&\ddots& \\
	 &&&&&&&& \\
	 1&-2&1 &0&\cdots &0&-1&2 &-1 \\
	 1&-2&1 &0&\cdots &0&0&-1 &2 \\
	 1&-2&1 &0& &\cdots&0&0 &-2 \\
	 \end{array}
       \right],
\end{eqnarray}
and $\mathbf{x}$ is the N/4 vector defined by
\begin{eqnarray}
	\mathbf{x}&=&[x_1, x_2, \cdots, x_{N/4}]^T\nonumber\\
&=&[\tilde{x}_1-\tilde{x}_{N/4+1}, \tilde{x}_2-\tilde{x}_{N/4+1}, \cdots, \tilde{x}_{N/4}-\tilde{x}_{N/4+1}]^T.
\label{eqn:ebar}
\end{eqnarray} 

Next, we discuss the matrix $\tilde{D}$. Let $\mathbf{d}_p =[D_{p,1}, D_{p,2},\cdots, D_{p,N/4}]\quad(p=1,2,\cdots, N/4+1)$ be the row vectors of $\tilde{D}$. This immediately leads
\begin{eqnarray}
	&&\mathbf{d}_1 = [1, 1, \cdots, 1]\\
	&&\mathbf{d}_{N/4+1}= [-1,1,\cdots,-1,1] 
\end{eqnarray}
Moreover, it can be easily derived that
\begin{eqnarray}
	\mathbf{d}_1 + \mathbf{d}_2 + \cdots + \mathbf{d}_{N/4} = [1,0,\cdots,1,0].
\end{eqnarray}
Therefore, we obtain the following relation
\begin{eqnarray}
	\mathbf{d}_{N/4+1} &=& \mathbf{d}_1-2(\mathbf{d}_1+\mathbf{d}_2+\cdots+\mathbf{d}_{N/4})\nonumber\\
		&=& -\mathbf{d}_1 - 2(\mathbf{d}_2+\mathbf{d}_3+\cdots+\mathbf{d}_{N/4})
		\label{eqn:d_n4_1}
\end{eqnarray}
or
\begin{eqnarray}
	D_{N/4+1,j} 
		= -D_{1,j} - 2(D_{2,j}+D_{3,j}+\cdots+D_{N/4,j}).\quad(j=1,2,\cdots,N/4)\nonumber\\
		\label{eqn:d_n4_1_component}
\end{eqnarray}

Using the relation (\ref{eqn:bbar2}), the element of vector (\ref{eqn:ebar}) can be rewritten as:
\begin{eqnarray}
	x_i &=& \tilde{x}_i - \tilde{x}_{N/4+1}\nonumber\\
	&=& (D_{i,1}-D_{N/4+1,1})b_{2} + (D_{i,2}-D_{N/4+1,2})b_{4}\nonumber\\ &+& \cdots + (D_{i,N/4}-D_{N/4+1,N/4}) b_{N/2}\nonumber\\
	&&\quad\quad(i=1,2,\cdots,N/4)
	\label{eqn:sub4}
\end{eqnarray}
Substituting (\ref{eqn:d_n4_1_component}) into (\ref{eqn:sub4}), we obtain
\begin{eqnarray}
	x_1 
	&=& 2(D_{1,1}+\cdots +D_{N/4,1})b_{2} + 2(D_{1,2}+\cdots+D_{N/4,2})b_{4} +\nonumber\\ &&\cdots + 2(D_{1,N/4}+\cdots D_{N/4,N/4})b_{N/2}\nonumber\\
	x_i &=&(D_{1,1}+2D_{2,1}+\cdots + 3D_{i,1}+\cdots+2D_{N/4,1})b_2\nonumber\\
	&&+(D_{1,2}+2D_{2,2}+\cdots + 3D_{i,2}+\cdots+2D_{N/4,2})b_4\nonumber\\
	&&+\cdots+(D_{1,N/4}+2D_{2,N/4}+\cdots + 3D_{i,N/4}+\cdots+2D_{N/4,N/4})b_{N/2}.\nonumber\\&&\quad \quad (i = 2,3,\cdots,N/4)
	\label{eqn:sub5}
\end{eqnarray}
The matrix form of Eq.~(\ref{eqn:sub5}) is 
\begin{eqnarray}
	\mathbf{x} = P_4 D \mathbf{b}_\mathrm{even},
	\label{eqn:ebar2}
\end{eqnarray}
where  $P_4$ is the $N/4 \times N/4$ matrix
\begin{eqnarray}
 P_4 =
  \left[\begin{array}{cccccccc}
   2&2&2&2&\ldots&2 &2 \\
   1&3&2&2&\ldots&2 &2 \\
   1&2&3&2&\ldots&2 &2 \\
   1&2&2&3&\ldots&2 &2 \\
   \vdots&\vdots&\vdots&\vdots&\ddots&\vdots &\vdots \\
   1&2&2&2&\ldots&3 &2 \\
   1&2&2&2&\ldots&2 &3 \\
  \end{array}\right],
  \label{eqn:p4}
\end{eqnarray}
and $D$ is the $N/4 \times N/4$ matrix 
constructed by removing the $(N/4+1)$-th row from $\tilde{D}$.

Finally, we obtain the transformed equation of (\ref{eqn:equation_for_symmetry_even_matrix}) by combining (\ref{eqn:sub2}) and (\ref{eqn:ebar2}),
\begin{eqnarray}
	\frac{1}{8}M_2 P_4 D \mathbf{b}_\mathrm{even} = \mathbf{g}.
	\label{eqn:trans_equation_for_symmetry_even_matrix}
\end{eqnarray}

As to the matrix $M_2$, the following lemma holds:
\begin{lemma}
	$M_2$ is regular matrix.
	\label{lemma:regularity_m2}
\end{lemma}
\begin{proof}
Consider the $N/4 \times N/4$ lower triangular matrix $P_2$ as follows
\begin{eqnarray}
 P_2 =
  \left[\begin{array}{cccccccc}
   1&0&0&0&\ldots&0 &0 &0 \\
   0&1&0&0&\ldots&0 &0 &0 \\
   0&2&1&0&\ldots&0 &0 &0 \\
   0&3&2&1&\ldots&0 &0 &0 \\
   \vdots&\vdots&\vdots&\vdots&\ddots&0 &0 &0 \\
   0&N/4-3&N/4-4&N/4-5&\ldots&1 &0 &0 \\
   0&N/4-2&N/4-3&N/4-4&\ldots&2 &1 &0 \\
   0&-2(N/4-3)&-2(N/4-4)&-2(N/4-5)&\ldots&-2 &0 &1 \\
  \end{array}\right].\nonumber\\
\end{eqnarray}

All diagonal components of the matrix $P_2$ are nonzero. Therefore $P_2$ is regular.
Applying $P_2$ to $M_2$, we obtain
\begin{eqnarray}
 P_2 M_2 = \left[
	\begin{array}{ccccccccc}
	 3&-4&1 &0 & & &\cdots &0 \\
	 a_{21}&a_{22}&a_{23}&-1 &0 & & \cdots &0 \\
	 a_{31}&a_{32}&a_{33}& 0 &-1&0& \cdots&0 \\
	 a_{41}&a_{42}&a_{33}& 0 &0&-1&\cdots&0 \\
	 \vdots&\vdots&&&&&\ddots& \\
	 a_{N/4-3,1}&a_{N/4-2,2}&a_{N/4-2,3}&0&\cdots &0&0&-1 \\
	 a_{N/4-1,1}&a_{N/4-1,2}&a_{N/4-1,3}&0& &\cdots&0 &0\\
	 a_{N/4,1}&a_{N/4,2}&a_{N/4,3}&0& &\cdots &0 &0\\
	\end{array}
       \right],
\end{eqnarray}
where
\begin{eqnarray}
 &&a_{n,1}=\left\{
  \begin{array}{ll}
   n(n-1)/2&n=2,3,\ldots,N/4-1\\
   -(N^2-20N+80)/16, &n=N/4
  \end{array}
 \right.\\
 &&a_{n,2}=\left\{
  \begin{array}{ll}
   1-n^2&n=2,3,\ldots,N/4-1\\
   (N^2-16N+32)/8,&n=N/4
  \end{array}
 \right.\\
 &&a_{n,3}=\left\{
  \begin{array}{ll}
   n(n+1)/2&n=2,3,\ldots,N/4-1\\
   -(N^2-12N+16)/16.&n=N/4
  \end{array}
 \right.
\end{eqnarray}
Finally, we consider the $N/4 \times N/4$ regular matrix $P_3$ as follows:
\begin{eqnarray}
 P_3 =\left[
\begin{array}{cccccccc}
 N/16&0 &0&0&\cdots &0 &1-N/8 &-N/16 \\
 0     &1 &0&0&\cdots&0 &0 &0 \\
 0&0 &1 &0 &\cdots &0 &0 &0 \\
 \vdots&\vdots & & &\ddots & &\vdots &\vdots \\
 0&0 &0 &0 &\cdots &1 &0 &0 \\
 \frac{(N-4)^2}{16N}&0 &0 &0&\cdots&0 &1+2/N-N/8 &1/N-N/16 \\
 \frac{(N-8)^2}{16N}&0 &0 &0 &\cdots&0&1+8/N-N/8 &4/N-N/16 \\
	\end{array}	
       \right]
\end{eqnarray}
Applying $P_3$ to $P_2 M_2$, we obtain
\begin{eqnarray}
 P_3 P_2 M_2 = \left[
	\begin{array}{ccccccccc}
	 1&0&0 &0 & & &\cdots &0 \\
	 a_{21}&a_{22}&a_{33}&-1 &0 & & \cdots &0 \\
	 a_{31}&a_{32}&a_{43}& 0 &-1&0& \cdots&0 \\
	 a_{41}&a_{42}&a_{53}& 0 &0&-1&\cdots&0 \\
	 \vdots&\vdots&&&&&\ddots& \\
	 a_{N/4-2,1}&a_{N/4-2,2}&a_{N/4-2,3}&0&\cdots &0&0&-1 \\
	 0&1&0&0& &\cdots&0 &0\\
	 0&0&1&0& &\cdots &0 &0\\
	\end{array}
       \right].
\end{eqnarray}
We can eliminate the components $a_{i,j}, 2\le i \le N/4-2, 1 \le j \le 3$ by elementary operations of adding the first, $N/4-1$
and $N/4$ rows. It is found that the rank of $P_3 P_2 M_2$ is
$N/4$. Since $P_2$ and $P_3$ is the regular matrix, the rank of $M_2$ is
equal to $N/4$. Therefore $M_2$ is regular.
\end{proof}

As to equations $\psi^{(i,j,k,l)}_\mathrm{even}(\mathbf{b}_\mathrm{even})=\sum_{s=1}^{N/4}b_1{w_{2s-1}
 c_{(2s-1)i} c_{(2s-1)j}c_{(2s-1)k}c_{(2s-1)l}}$, $\forall (i,j,k,l)\in S_2$, the following proposition holds:

\begin{prop}
For a given solution $b_{2n-1} = b_{1}w_{2n-1}, n=1,2,\cdots N/4$ for the equations $\psi^{(i,j,k,l)}_\mathrm{odd}(\mathbf{b}_\mathrm{odd})=0$,$\forall (i,j,k,l)\in S_1$, equations
$\psi^{(i,j,k,l)}_\mathrm{even}(\mathbf{b}_\mathrm{even})=\sum_{s=1}^{N/4}b_1{w_{2s-1}
c_{(2s-1)i}c_{(2s-1)j}c_{(2s-1)k}c_{(2s-1)l}}$, $\forall (i,j,k,l)\in S_2$ are $N/4$ linearly
independent equations.
\label{prop:independent_even}
\end{prop}
\begin{proof}
The equation (\ref{eqn:equation_for_symmetry_even_matrix}) is equivalent to the transformed equation (\ref{eqn:trans_equation_for_symmetry_even_matrix}).
$M_2$ is a $N/4 \times N/4$ regular matrix from Lemma \ref{lemma:regularity_m2}.
Since $P_4D$ is the $N/4 \times N/4$ regular matrix(see \ref{sec:regularity_D}), $M_2 P_4 D$ is the $N/4 \times N/4$ regular matrix. Therefore, $\psi^{(i,j,k,l)}_\mathrm{even}(\mathbf{b}_\mathrm{even})=\sum_{s=1}^{N/4}b_1{w_{2s-1}c_{(2s-1)i}c_{(2s-1)j}c_{(2s-1)k}c_{(2s-1)l}}$, $\forall (i,j,k,l)\in S_2$ are $N/4$ linearly independent equations. 
\end{proof}
		
The solution of Eq.~(\ref{eqn:trans_equation_for_symmetry_even_matrix}) can be written as 
$\mathbf{b}_\mathrm{even}=8(P_4D)^{-1}	M_2^{-1}\mathbf{g}$, which leads to
$b_{2s} = 8[(P_4 D)^{-1} M_2^{-1}]_{s,1}Wb_1\quad (s=1,2,\cdots, N/4)$. 
Therefore, it can be also written in the form
\begin{eqnarray}
 b_{2s} = w_{2s}b_1,\quad (s = 1,2,\cdots,N/4)
  \label{eqn:solution_even}
\end{eqnarray}
where $w_{2s} = 8[(P_4 D)^{(-1)} M_2^{-1}]_{s,1}W$.

For converting Eq.~(\ref{eqn:trans_equation_for_symmetry_even_matrix}) into a simple form, we introduce the $N/4\times N/4$ regular matrices $Q_2$, $Q_3$ and $Q_4$ as follows:
\begin{eqnarray}
 Q_2 =
  \left[
   \begin{array}{ccccccc}
    1&0 &0 &\cdots &0 &0&0 \\
    0&1 &0 &\cdots &0 &0&0\\
    0&0 &1 &\cdots &0 &0&0 \\
    \vdots&\vdots&\vdots &\ddots &\vdots &\vdots&\vdots \\
    0&0 &0 &\cdots &1 &0&0\\
    0&2 &2 &\cdots &2&1 &0 \\
    0&0 &0 &\cdots &0&2 &1 \\
   \end{array}
  \right]
\end{eqnarray}
\begin{eqnarray}
 Q_3 = 
 \left[
   \begin{array}{ccccccc}
    1&0 &0 &\cdots &0 &0&0 \\
    0&1 &0 &\cdots &0 &0&0 \\
    0&0 &1 &\cdots &0 &0&0 \\
    \vdots&\vdots&\vdots &\ddots &\vdots &\vdots&\vdots \\
    0&0 &0 &\cdots &1 &0 &0\\
    0&0 &0 &\cdots &0&\frac{24}{N(N+2)} &-\frac{N-4}{4(N+2)} \\
    0&0 &0 &\cdots &0&\frac{48(N-4)}{N(N^2-4)}&-\frac{N^2-12N+8}{2(N^2-4)} \\
   \end{array}
  \right]
\end{eqnarray}
\begin{eqnarray}
 Q_4 =
  \left[
   \begin{array}{ccccccc}
    1&0 &0 &\cdots &0 &0&0 \\
    0&-1 &0 &\cdots &0 &a_{22}&a_{23} \\
    0&0 &-1 &\cdots &0 &a_{32}&a_{33} \\
    \vdots&\vdots&\vdots &\ddots &\vdots &\vdots&\vdots \\
    0&0 &0 &\cdots &-1 &a_{N/4-2,2} &a_{N/4-2,3}\\
    0&0 &0 &\cdots &0&1 &0 \\
    0&0 &0 &\cdots &0&0 &1 \\
   \end{array}
  \right]
\end{eqnarray}
Applying $Q_4 Q_3 Q_2 P_2$ to (\ref{eqn:trans_equation_for_symmetry_even_matrix}), we
obtain
\begin{eqnarray}
 \frac{1}{8}
  \left[
   \begin{array}{cccccccc}
    3&-4 &1 &0 &0 &0 &\cdots&0 \\
    h_{N}(3)&0 &0 &1 &0 &0 &\cdots&0 \\
    h_{N}(4)&0 &0 &0 &1 &0&\cdots &0 \\
    \vdots&\vdots &\vdots &\vdots &\vdots &&\ddots &\vdots \\
    h_{N}(N/4-1)&0 &0 &0 &0 &\cdots &0 &1 \\
    h_{N}(1)&1 &0 &0 &0 &0 &\cdots &0\\
    h_{N}(2)&0 &1 &0 &0 &0 &\cdots &0 \\
   \end{array}
				\right]D\mathbf{b}_\mathrm{even} = \mathbf{g}.
  \label{eqn:equation_for_symmetry_odd_matrix2}
\end{eqnarray}
where $h_m(n)=-1 -\frac{24n^2-12nm}{m^2-4}$.

Substituting Eq.~(\ref{eqn:solution_even}) into
Eq.~(\ref{eqn:equation_for_symmetry_odd_matrix2}), the following relation holds:
\begin{eqnarray}
 \frac{1}{8}\sum_{s=1}^{N/4}[3-4c_{4s}+c_{8s}]w_{2s} = W,
  \label{eqn:relation_even0}
\end{eqnarray}
and
\begin{eqnarray}
 \frac{1}{8}\sum_{s=1}^{N/4}[h_{N}(n) +c_{4ns}]w_{2s} = 0.\quad
  (n=1,2,\cdots,N/4-1)
 \label{eqn:relation_even}
\end{eqnarray}
It can be shown that Eq.(\ref{eqn:relation_even}) also holds for $n=0$ and $n=N/2$.
Define new variables $n' = N/2-n$. Substituting $n=N/2-n'$ into
Eq.(\ref{eqn:relation_even}), we obtain
\begin{eqnarray}
\frac{1}{8}\sum_{s=1}^{N/4}[h_{N}(n') +c_{4n's}]w_{2s}=0
\end{eqnarray}
for $n' = N/4+1, N/4+2,\cdots,N/2-1$. Therefore,
Eq.(\ref{eqn:relation_even}) holds for the extended range $n=0,1,\cdots, N/2$.

Substituting (\ref{eqn:relation_even}) for $n=1,2$ into (\ref{eqn:relation_even0}),
we obtain the relation
\begin{eqnarray}
 \frac{3N}{N^2-4}\sum_{s=1}^{N/4}w_{2s} = W.
  \label{eqn:relation_even_odd}
\end{eqnarray}

We consider the equations $\psi_{\mathrm{even}}^{(i,j,k,l)}(\mathbf{b}_\mathrm{even})=\sum_{s=1}^{N/4}b_1{w_{2s-1}
 c_{(2s-1)i} c_{(2s-1)j}c_{(2s-1)k}c_{(2s-1)l}}$, $\forall (i,j,k,l)\in T_2$.

\begin{prop}
For the given solutions $b_{2n-1} = b_{1}w_{2n-1}\quad (n=1,2,\cdots N/4)$ for the equations $\psi^{(i,j,k,l)}_\mathrm{odd}(\mathbf{b}_\mathrm{odd})=0$,$\forall (i,j,k,l)\in S_1$, 
 the nontrivial solution of equations\\ $\psi^{(i,j,k,l)}_\mathrm{even}(\mathbf{b}_\mathrm{even})=\sum_{s=1}^{N/4}b_1{w_{2s-1}
 c_{(2s-1)i} c_{(2s-1)j}c_{(2s-1)k}c_{(2s-1)l}}$, $\forall (i,j,k,l)\in S_2$ also solves the other equations in $T_2$.
 \label{prop:solve_on_t2}
\end{prop}
\begin{proof}
The set $T_2$ can be expressed by $(i,j,k,l)=(a+b-c,c,N/2-b,N/2-a)$, where $a$, $b$ and $c$ are integers. 
Substituting $i=a+b-c, j=c , k=N/2-b, l=N/2-a$ into
Eq.(\ref{eqn:s2_lhs}) which is LHS of Eq.~(\ref{eqn:equation_for_symmetric_s2}), we obtain
\begin{eqnarray}
 &&\psi^{(a+b-c,c,N/2-b,N/2-a)}_\mathrm{even}(\mathbf{b}_\mathrm{even})\nonumber\\
  &&=
  \sum_{s=1}^{N/4}b_{2s}s_{2(a+b-c)s}s_{2cs}s_{(N-2b)s}s_{(N-2a)s}\nonumber\\
 &&=\frac{1}{8}\sum_{s=1}^{N/4}\left[1+c_{4(a+b)s}-c_{4(c-a-b)s}-c_{4as}-c_{4bs}-c_{4cs}+c_{4(c-a)s}+c_{4(c-b)s}\right]b_{2s}.\nonumber\\
 \label{eqn:lhs_t2}
\end{eqnarray}
Since $i\le j\le k\le l$, it is found that $(a+b)/2 \le c$ and $a\le b$. 
Therefore, the relation $a \le (a+b)/2 \le c$ holds. Then we have
$c-a \ge 0$. Moreover, we have $a,b \ge 1$ since $N/2-a\le N/2-1$ and $N/2-b\le N/2-1$. Therefore, we have $1\le a\le c$.
In order to keep the index $\alpha$ of $c_\alpha$ positive or zero in the 3rd, 7th and last terms in Eq.(\ref{eqn:lhs_t2}),
three cases of RHS of Eq.~(\ref{eqn:lhs_t2}) should be considered.
\begin{enumerate}
 \item $a+b-c < 0$, $c-a\ge 0$ and $c-b\ge0$
       \begin{eqnarray}
	L_1&=&\frac{1}{8}\sum_{s=1}^{N/4}\left[1+c_{4(a+b)s}-c_{4(c-a-b)s}-c_{4as}-c_{4bs}-c_{4cs}\right.\nonumber\\
	&&\left.+c_{4(c-a)s}+c_{4(c-b)s}\right]b_{2s}\nonumber\\
	 \label{eqn:equation_lhs_s2_1}
       \end{eqnarray}
 \item $a+b-c \ge 0$, $c-a\ge 0$ and $c-b\ge0$
       \begin{eqnarray}
	L_2&=&\frac{1}{8}\sum_{s=1}^{N/4}\left[1+c_{4(a+b)s}-c_{4(a+b-c)s}-c_{4as}-c_{4bs}-c_{4cs}\right.\nonumber\\
	&&\left.+c_{4(c-a)s}+c_{4(c-b)s}\right]b_{2s}\nonumber\\
	 \label{eqn:equation_lhs_s2_2}
       \end{eqnarray}
 \item $a+b-c \ge 0$, $c-a\ge 0$ and $c-b<0$
       \begin{eqnarray}
	L_3&=&\frac{1}{8}\sum_{s=1}^{N/4}\left[1+c_{4(a+b)s}-c_{4(a+b-c)s}-c_{4as}-c_{4bs}-c_{4cs}\right.\nonumber\\
	&&\left.+c_{4(c-a)s}+c_{4(b-c)s}\right]b_{2s}\nonumber\\
	 \label{eqn:equation_lhs_s2_3}
       \end{eqnarray}
\end{enumerate}
Using (\ref{eqn:solution_even}) and (\ref{eqn:relation_even}),
Eq.(\ref{eqn:equation_lhs_s2_1})-(\ref{eqn:equation_lhs_s2_3}) are
simplified as follows:
\begin{eqnarray}
 L_1&=&\frac{b_1}{8}\sum_{s=1}^{N/4}\left[1-h_{N}(a+b)+h_{N}(c-a-b)+h_{N}(a)+h_{N}(b)+h_{N}(c)\right.\nonumber\\
 &&\left.-h_{N}(c-a)-h_{N}(c-b)\right]w_{2s}\nonumber\\
 &=&0.
  \label{eqn:equation_lhs_s2_1_2}
\end{eqnarray}
\begin{eqnarray}
 L_2&=&\frac{b_1}{8}\sum_{s=1}^{N/4}\left[1-h_{N}(a+b)+h_{N}(a+b-c)+h_{N}(a)+h_{N}(b)+h_{N}(c)\right.\nonumber\\
 &&\left.-h_{N}(c-a)-h_{N}(c-b)\right]w_{2s}\nonumber\\
 &=&\frac{3Nb_1}{N^2-4}(a+b-c)\sum_{s=1}^{N/4}w_{2s}.
  \label{eqn:equation_lhs_s2_2_2}
\end{eqnarray}
\begin{eqnarray}
 L_3&=&\frac{b_1}{8}\sum_{s=1}^{N/4}\left[1-h_{N}(a+b)+h_{N}(a+b-c)+h_{N}(a)+h_{N}(b)+h_{N}(c)\right.\nonumber\\
 &&\left.-h_{N}(c-a)-h_{N}(b-c)\right]w_{2s}\nonumber\\
 &=&\frac{3Nb_1}{N^2-4}a\sum_{s=1}^{N/4}w_{2s}.
  \label{eqn:equation_lhs_s2_3_2}
\end{eqnarray}

Substituting $i=a+b-c, j=c , k=N/2-b, l=N/2-a$ into
Eq.(\ref{eqn:s2_rhs}) which is the RHS of Eq.(\ref{eqn:equation_for_symmetric_s2}), we obtain
\begin{eqnarray}
 \mathcal{R}^{(a+b-c,c,N/2-b,N/2-a)}&=&\sum_{s=1}^{N/4}\left[c_{(2s-1)(a+b-c)}c_{(2s-1)c}c_{(2s-1)(N/2-b)}c_{(2s-1)(N/2-a)}\right]b_1 w_{2s-1}\nonumber\\
 &=&\frac{1}{8}\sum_{s=1}^{N/4}\left[-1-c_{2(2s-1)(a+b)}+c_{2(2s-1)a}+c_{2(2s-1)b}-c_{2(2s-1)c}\right.\nonumber\\
&&\left.-c_{2(2s-1)(a+b-c)}+c_{2(2s-1)(a-c)}+c_{2(2s-1)(b-c)}\right]b_1 w_{2s-1}.\nonumber\\
\end{eqnarray}
As same as the LHS, three cases should be considered in order to keep
the index $\alpha$ of $c_{\alpha}$ positive or zero. In each case, the
equation is simplified by using the relation (\ref{eqn:relation_odd}).
\begin{enumerate}
 \item $a+b-c < 0$, $c-a\ge 0$ and $c-b\ge0$
       \begin{eqnarray}
	R_1&=&\frac{b_1}{8}\sum_{s=1}^{N/4}\left[-1-c_{2(2s-1)(a+b)}+c_{2(2s-1)a}+c_{2(2s-1)b}\right.\nonumber\\
	&&\left.
	-c_{2(2s-1)c}-c_{2(2s-1)(c-a-b)}+c_{2(2s-1)(c-a)}+c_{2(2s-1)(c-b)}\right]w_{2s-1}\nonumber\\
	&=&\frac{b_1}{8}\sum_{s=1}^{N/4}\left[-1+f_{N/4}(a+b)-f_{N/4}(a)-f_{N/4}(b)\right.\nonumber\\
	&&\left.+f_{N/4}(c)+f_{N/4}(c-a-b)-f_{N/4}(c-a)-f_{N/4}(c-b)\right]w_{2s-1}\nonumber\\
	&=&0.
	 \label{eqn:equation_rhs_s2_1}
       \end{eqnarray}
 \item $a+b-c \ge 0$, $c-a\ge 0$ and $c-b\ge0$
       \begin{eqnarray}
	R_2&=&\frac{b_1}{8}\sum_{s=1}^{N/4}\left[-1-c_{2(2s-1)(a+b)}+c_{2(2s-1)a}+c_{2(2s-1)b}\right.\nonumber\\
&&\left.-c_{2(2s-1)c}-c_{2(2s-1)(a+b-c)}+c_{2(2s-1)(c-a)}+c_{2(2s-1)(c-b)}\right]w_{2s-1}.\nonumber\\
	&=&\frac{b_1}{8}\sum_{s=1}^{N/4}\left[-1+f_{N/4}(a+b)-f_{N/4}(a)-f_{N/4}(b)\right.\nonumber\\
	&&\left.+f_{N/4}(c)+f_{N/4}(a+b-c)-f_{N/4}(c-a)-f_{N/4}(c-b)\right]w_{2s-1}\nonumber\\
	&=&W{b_1}(a+b-c).
	 \label{eqn:equation_rhs_s2_2}
       \end{eqnarray}
 \item $a+b-c \ge 0$, $c-a\ge 0$ and $c-b<0$
       \begin{eqnarray}
	R_3&=&\frac{b_1}{8}\sum_{s=1}^{N/4}\left[-1-c_{2(2s-1)(a+b)}+c_{2(2s-1)a}+c_{2(2s-1)b}\right.\nonumber\\
&&\left.-c_{2(2s-1)c}-c_{2(2s-1)(a+b-c)}+c_{2(2s-1)(c-a)}+c_{2(2s-1)(b-c)}\right]w_{2s-1}\nonumber\\
	&=&\frac{b_1}{8}\sum_{s=1}^{N/4}\left[-1+f_{N/4}(a+b)-f_{N/4}(a)-f_{N/4}(b)\right.\nonumber\\
	&&\left.+f_{N/4}(c)+f_{N/4}(a+b-c)-f_{N/4}(c-a)-f_{N/4}(b-c)\right]w_{2s-1}\nonumber\\
	&=&W{b_1}a.
	 \label{eqn:equation_rhs_s2_3}
       \end{eqnarray}
\end{enumerate}

From the relation (\ref{eqn:relation_even_odd}), it follows that $L_i - R_i =0\quad (i=1,2,3)$.
This indicates that the solution $\mathbf{b}_\mathrm{even}$ of the equation $\psi^{(i,j,k,l)}_\mathrm{even}(\mathbf{b}_\mathrm{even})=\sum_{s=1}^{N/4}b_1{w_{2s-1}
 c_{(2s-1)i} c_{(2s-1)j}c_{(2s-1)k}c_{(2s-1)l}}$, $\forall (i,j,k,l)\in S_2$ also solves the other equations in $T_2$ when $b_{2n-1} = b_{1}w_{2n-1}\quad (n=1,2,\cdots, N/4)$ is the solution of the equations $\psi^{(i,j,k,l)}_\mathrm{odd}(\mathbf{b}_\mathrm{odd})=0$,$\forall (i,j,k,l)\in S_1$.
 \end{proof}

\section{Proof of Main Results}\label{sec:proof}
\begin{proof}[{\bf Proof of Lemma \ref{theo:main1}}]
	From Proposition \ref{prop:independent_odd} and Proposition \ref{prop:independent_even}, the equations $\psi^{(i,j,k,l)}(\mathbf{b})=0,\forall (i,j,k,l)\in S_1 \cup S_2$ are $N-1$ linearly independent equations. Therefore, they have a nontrivial solution $\mathbf{b}\ne 0$ for given $b_1$. From Proposition \ref{prop:solve_on_t1} and Proposition \ref{prop:solve_on_t2}, the nontrivial solution also solves the equations $\psi^{(i,j,k,l)}(\mathbf{b})=0,\forall (i,j,k,l)\in T_1 \cup T_2 = S$. From Lemma \ref{lemma:minus}, the solution also solves the equations for $S_0$.
\end{proof}

\begin{proof}[{\bf Proof of Theorem \ref{theo:main2}}]
	From Lemma \ref{theo:main1}, a trivial solution $\mathbf{b}$ for $\psi^{(i,j,k,l)}(\mathbf{b})=0,\forall (i,j,k,l)\in S_1 \cup S_2$ also solves $\psi^{(i,j,k,l)}(\mathbf{b})=0,\forall (i,j,k,l)\in S_0$. Therefore, the asymmetric part $\Psi_\mathrm{a}$ of Hamiltonian (\ref{eqn:proposed_model}) vanishes for such $\mathbf{b}$. Therefore, from Definition \ref{def:symmetric}, the system (\ref{eqn:proposed_model}) with the solution $\mathbf{b}$ is a symmetric lattice.
\end{proof}

\section*{Acknowledgements}	
	The first author (Y.D.) was partially supported by a Grant-in-Aid for Scientific Research (C), No.~19K12003 from Japan Society for the Promotion of Science (JSPS).  The authors were supported by a
 Grant-in-Aid for Scientific Research (C), No.~19K03654
 from Japan Society for the Promotion of Science (JSPS).

\appendix
\section{Estimating velocity $v_{DB}$ of approximated traveling DB}\label{sec:estimate_vdb}
The map $\mathcal{T}_\lambda$ with an arbitrary $\lambda$ defined by Eq.(\ref{eqn:transform_T}) represents the arbitrary space shifting and sign-inverting transformation. This map corresponds to the rotation by the angle $-m\lambda$ in each $U_m$ component on the complex plane. Therefore, we can estimate the distance that a DB travels in the lattice from the angle that the component $U_m$ rotates in the complex plane.

Figure \ref{fig:trajectory_complex} shows the trajectory of component $U_m$ for an approximated traveling DB on the complex plane. The trajectory can be decomposed into the fast vibration corresponding to the internal vibration and the slow rotation corresponding to the propagation of a traveling DB. 
From the definition of map (\ref{eqn:transform_T}), the component $U_m$ rotates $-2m\pi$ during the DB propagated $N$ lattice spacing. Therefore, the rotation angle $\theta_1=-2m\pi/N$ of the component $U_m$ corresponds to an one-lattice-space propagation of the traveling DB.

We can estimate $v_\mathrm{DB}$ by the following steps:
\begin{enumerate}
\item Perform the numerical integration of (\ref{eqn:motion_tascl_q}) and (\ref{eqn:motion_tascl}) for the traveling DB with a certain small perturbation $\delta l_0$  (\ref{eqn:variable_transformation_to_dp}). Transform the obtained temporal evolution $q_n$ into $U_m$.
\item Find the time $t_1$ and $t_3$ of the first and third extreama of $|U_m|$.
\item Estimate the rotation angle $|\Delta \theta|$ between $t_1$ and $t_3$. This corresponds to the rotation angle one internal vibration: 
\begin{eqnarray}
  |\Delta \theta | = \arg U_m(t_3)- \arg U_m(t_1),
\end{eqnarray}
where $\arg$ indicates the argument of complex numbers.
\item Calculate the velocity $v_\mathrm{DB}$ by
\begin{eqnarray}
	v_\mathrm{DB} = \left|\frac{\Delta \theta}{\theta_1}\right|=\left|\frac{N\Delta \theta}{2m\pi}\right|.
\end{eqnarray}
\end{enumerate}

As is shown in Fig.~\ref{fig:vdb_in_perturbed_db}, the velocity $v_\mathrm{DB}$ of approximated traveling DB is precisely proportional to $\delta l$. Therefore, we obtain the relation $v_\mathrm{DB} = K\delta l$. The coefficient $K$ can be calculated from one pair of $\delta l_0$ and $v_\mathrm{DB}$ for a certain $T_\mathrm{DB}$ following the above procedure. Finally, the $\delta l$ for desired $v_\mathrm{DB}$ can be obtained by $\delta l = v_\mathrm{DB}/K$.

\begin{figure}
\centering
\begin{tabular}{c}
	\includegraphics{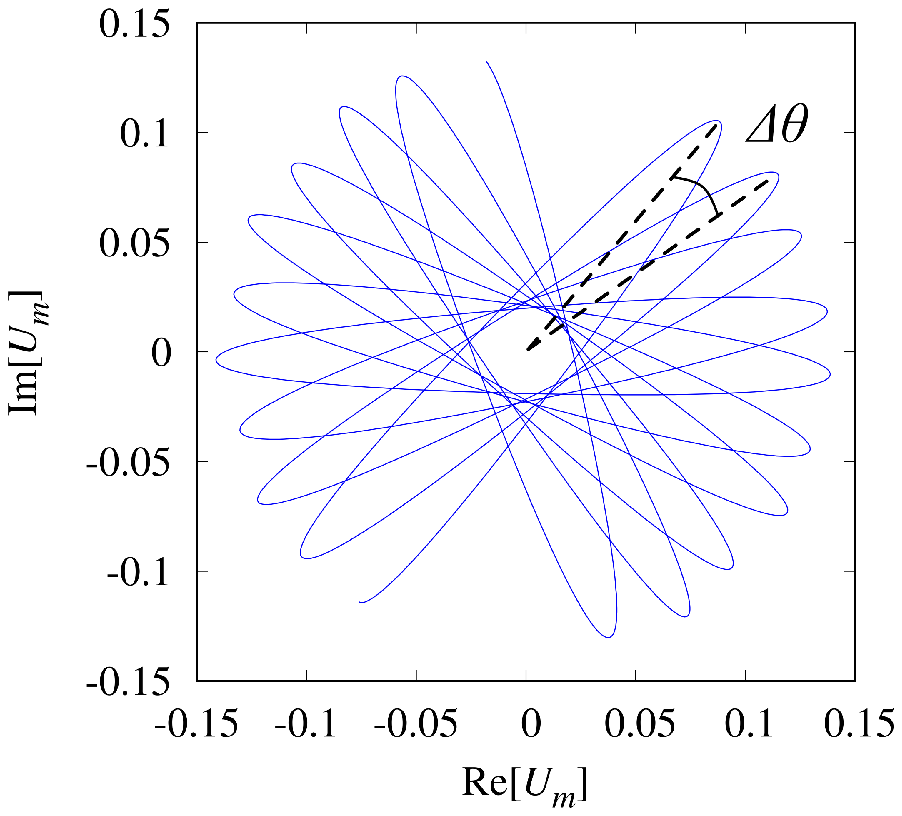}\\
\end{tabular}
	\caption{Trajectory of the component $U_m$ in the complex coordinate. $\Delta \theta$ indicates the change of angle in one internal vibration. Numerical results are the cases that $T_\mathrm{DB}=2, \delta l=0.02$, and $m=20$.}
	\label{fig:trajectory_complex}
\end{figure}

\section{Proof of regularity of $A$}\label{sec:regularity_C}
\setcounter{theo}{0}
$A$ is the $N/4 \times N/4$ matrix whose element is given by $A_{pq} =
\cos{\frac{2\pi}{N}(p-1)(2q-1)}$. Consider the $N/4\times N/4$ matrix
$\bar{A}$ as follows,
\begin{eqnarray}
 \bar{A} = \frac{8}{N}
  \left[
   \begin{array}{ccccc}
    \frac{1}{2}A_{11}&A_{21}&A_{31}&\cdots&A_{\frac{N}{4}1}\\
    \frac{1}{2}A_{12}&A_{22}&A_{32}&\cdots&A_{\frac{N}{4}2}\\
    \frac{1}{2}A_{13}&A_{23}&A_{33}&\cdots&A_{\frac{N}{4}3}\\
    \vdots&\vdots&\vdots&\ddots&\vdots\\
    \frac{1}{2}A_{1\frac{N}{4}}&A_{2\frac{N}{4}}&A_{3\frac{N}{4}}&\cdots&A_{\frac{N}{4}\frac{N}{4}}\\
   \end{array}
  \right]
\end{eqnarray}
It can be easily shown that $A\bar{A} = \bar{A}A = I$. Therefore,
$A^{-1} = \bar{A}$ and this means that $A$ is regular.

\section{Proof of regularity of $P_4D$}\label{sec:regularity_D}
$P_4$ is the $N/4 \times N/4$ matrix given in (\ref{eqn:p4}) and $D$ is the $N/4 \times N/4$ matrix that element is given by $D_{pq} = \cos{\frac{4(p-1)q\pi}{N}}$. Consider the $N/4\times N/4$ matrix $\bar{D}$ as follows:
\begin{eqnarray}
	\bar{D} = \frac{8}{N}
	\left[
	\begin{array}{ccccc}
		\frac{1}{2}D_{11}&D_{21}&D_{31}&\cdots&D_{N/4,1}\\
		\frac{1}{2}D_{12}&D_{22}&D_{32}&\cdots&D_{N/4,2}\\
		\frac{1}{2}D_{13}&D_{23}&D_{33}&\cdots&D_{N/4,3}\\
	    \vdots&\vdots&\vdots&\ddots&\vdots\\
		\frac{1}{2}D_{1,N/4-1}&D_{2,N/4-1}&D_{3,N/4-1}&\cdots&D_{N/4,N/4-1}\\
		\frac{1}{4}D_{1,N/4}&\frac{1}{2}D_{2,N/4}&\frac{1}{2}D_{3,N/4}&\cdots&\frac{1}{2}D_{N/4,N/4}\\
	\end{array}
	\right].
\end{eqnarray}
It can be shown that $\bar{D}P_4 D = P_4 D \bar{D} = I$. Therefore, $(P_4 D)^{-1} = \bar{D}$ and this means that $P_4 D$ is regular.

\section{Derivation of explicit solution}\label{sec:explicit_solution}
At first, we consider Eq.~(\ref{eqn:equation_for_symmetry_odd_matrix}). 
Applying the matrix $P_1$ from left side, we obtain
\begin{eqnarray}
	P_1 M_1 A \mathbf{b}_\mathrm{odd} = 0
	\label{eqn:p1m1c}
\end{eqnarray}
The matrix $P_1 M_1$ is given in (\ref{eqn:p1m1}). 
We introduce the $N/4$-vector $\mathbf{B}$ as follows:
\begin{eqnarray}
\mathbf{B} = [B_1, B_2, \cdots, B_{N/4}]^T =A \mathbf{b}_\mathrm{odd}.
\label{eqn:def_vecB}
\end{eqnarray}
The equation (\ref{eqn:p1m1c}) is rewritten to
\begin{eqnarray}
	P_1 M_1 \mathbf{B} = 0,
\end{eqnarray}
and this equation leads to the relation
\begin{eqnarray}
	B_n  = -f_{N/4}(n-1)B_1.\quad (n=2,3,\cdots,N/4)
	\label{def:vectorB}
\end{eqnarray}
Using (\ref{def:vectorB}), $\mathbf{B}$ can be rewritten to
\begin{eqnarray}
	\mathbf{B} = B_1 [1, -f_{N/4}(1), -f_{N/4}(2), \cdots, -f_{N/4}(N/4-1)]^T.
	\label{def:vectorB2}
\end{eqnarray}
Using (\ref{eqn:p1m1c}) and (\ref{def:vectorB2}), and considering \ref{sec:regularity_D}, we obtain
\begin{eqnarray}
	\mathbf{b}_\mathrm{odd}&=& B_1 A^{-1}[1, -f_{N/4}(1), -f_{N/4}(2), \cdots, -f_{N/4}(N/4-1)]^T\nonumber\\
	&=& B_1 \bar{A}[1, -f_{N/4}(1), -f_{N/4}(2), \cdots, -f_{N/4}(N/4-1)]^T\nonumber\\
	&=& \frac{8B_1}{N}
  \left[
   \begin{array}{ccccc}
    \frac{1}{2}A_{11}&A_{21}&A_{31}&\cdots&A_{\frac{N}{4}1}\\
    \frac{1}{2}A_{12}&A_{22}&A_{32}&\cdots&A_{\frac{N}{4}2}\\
    \frac{1}{2}A_{13}&A_{23}&A_{33}&\cdots&A_{\frac{N}{4}3}\\
    \vdots&\vdots&\vdots&\ddots&\vdots\\
    \frac{1}{2}A_{1\frac{N}{4}}&A_{2\frac{N}{4}}&A_{3\frac{N}{4}}&\cdots&A_{\frac{N}{4}\frac{N}{4}}\\
   \end{array}
  \right]
    \left[
   \begin{array}{c}
    1\\
    -f_{N/4}(1)\\
    -f_{N/4}(2)\\
    \vdots\\
    -f_{N/4}(N/4-1)
   \end{array}
  \right].\nonumber\\
\end{eqnarray}
Therefore, we obtain
\begin{eqnarray}
	b_{2s-1} &=& \frac{8B_1}{N}\left[\frac{1}{2}A_{1s}-\sum_{k=1}^{N/4-1}{A_{k+1,s}f_{N/4}(k)}\right]\nonumber\\
	&=&\frac{8B_1}{N}\left[\frac{1}{2}-\sum_{k=1}^{N/4-1}\cos{\frac{2\pi k(2s-1)}{N}f_{N/4}(k)}\right]\nonumber\\
	&=&\frac{8B_1}{N^2}\frac{1}{\sin^2\left((2s-1)\pi/N\right)}.\quad(s=1,2,\cdots,N/4)
	\label{eqn:solution_explicit_odd1}
\end{eqnarray}
Moreover, we obtain
\begin{eqnarray}
	w_{2s-1}=\frac{b_{2s-1}}{b_1} = \frac{\sin^2{(\pi/N)}}{\sin^2{((2s-1)\pi/N)}}
	\quad(s=1,2,\cdots,N/4).
	\label{sol:w_odd}
\end{eqnarray}

Following additional calculations are performed for the further discussion. From the first row in Eq.~(\ref{eqn:def_vecB}), the following relation holds:
\begin{eqnarray}
B_1 = \sum_{s'=1}^{N/4}{b_{2s'-1}} = b_1\sum_{s'=1}^{N/4}{w_{2s'-1}}=b_1 N W,
\label{eqn:def_B1}
\end{eqnarray}
In the last equality, we use Eq.~(\ref{eqn:definition_R}).
It is also obtained from Eq.~(\ref{eqn:solution_explicit_odd1}) by setting $s=1$.
\begin{eqnarray}
	B_1 = \frac{b_1}{8}N^2 \sin^2{\frac{\pi}{N}}.\label{eqn:b1b1}
\end{eqnarray}
Comparing Eq.(\ref{eqn:def_B1}) and (\ref{eqn:b1b1}), we obtain
\begin{eqnarray}
	W = \frac{N}{8}\sin^2{\frac{\pi}{N}}\label{eqn:def_W}
\end{eqnarray}

Next, we consider Eq.~(\ref{eqn:trans_equation_for_symmetry_even_matrix}). Applying $Q_4Q_3Q_2P_2$ from the left side, we obtain Eq.~(\ref{eqn:equation_for_symmetry_odd_matrix2}). We introduce $N/4$-vector: 
\begin{eqnarray}
	\mathbf{G}=[G_1, G_2, \cdots, G_{N/4}]^T = D\mathbf{b}_\mathrm{even}.
	\label{eqn:vectorG}
\end{eqnarray}
 Equation (\ref{eqn:equation_for_symmetry_odd_matrix2}) is rewritten to
\begin{eqnarray}
 \frac{1}{8}
  \left[
   \begin{array}{cccccccc}
    3&-4 &1 &0 &0 &0 &\cdots&0 \\
    h_{N}(3)&0 &0 &1 &0 &0 &\cdots&0 \\
    h_{N}(4)&0 &0 &0 &1 &0&\cdots &0 \\
    \vdots&\vdots &\vdots &\vdots &\vdots &&\ddots &\vdots \\
    h_{N}(N/4-1)&0 &0 &0 &0 &\cdots &0 &1 \\
    h_{N}(1)&1 &0 &0 &0 &0 &\cdots &0\\
    h_{N}(2)&0 &1 &0 &0 &0 &\cdots &0 \\
   \end{array}
				\right]\mathbf{G} = \mathbf{g},
\end{eqnarray}
and we obtain the relations
\begin{eqnarray}
	&&3G_1 -4G_2 + G_3 = 8Wb_1,\label{eqn:G1G2G2}\\
	&&G_n = -h_{N}(n-1)G_1\quad(n = 2,\cdots, N/4).
	\label{eqn:relationG}
\end{eqnarray}
Substituting (\ref{eqn:relationG}) into (\ref{eqn:G1G2G2}), we obtain
\begin{eqnarray}
	G_1 = \frac{N^2-4}{3N}Wb_1 = \frac{N^2-4}{24}b_1\sin^2{\frac{\pi}{N}},
	\label{eqn:g1b1}
\end{eqnarray} 
where we use Eq.~(\ref{eqn:def_W}) for the last equality.

Applying $P_4$ to (\ref{eqn:vectorG}) from the left side and substituting (\ref{eqn:relationG}), we obtain
\begin{eqnarray}
	P_4 D \mathbf{b}_\mathrm{even} = G_1P_4[1,-h_N(1),-h_N(2),\cdots,-h_N(N/4-1)]^T.
	\label{eqn:p4dbe.g1p4h}
\end{eqnarray}
Let $P_4[1,-h_N(1),-h_N(2),\cdots,-h_N(N/4-1)]^T = [H_1, H_2, \cdots, H_{N/4}]^T=\mathbf{H}$. Elements of $\mathbf{H}$ are given by
\begin{eqnarray}
	H_1 &=& 2(1-\sum_{n=1}^{N/4-1}{h_N(n))}=\frac{3N^2}{2(N^2-4)},\\
	H_m &=& H_1-1-h_N(m-1)
     =\frac{3(N-4m+4)^2}{2(N^2-4)}\quad(m=2,3,\cdots, N/4)\nonumber\\
\end{eqnarray}
Using the relation $(P_4 D)^{-1} = \bar{D}$ (see \ref{sec:regularity_D}), we obtain
\begin{eqnarray}
\mathbf{b}_\mathrm{even} &=& G_1\bar{D}\mathbf{H}\nonumber\\
&=&\frac{8G_1}{N}
	\left[
	\begin{array}{ccccc}
		\frac{1}{2}D_{11}&D_{21}&D_{31}&\cdots&D_{N/4,1}\\
		\frac{1}{2}D_{12}&D_{22}&D_{32}&\cdots&D_{N/4,2}\\
		\frac{1}{2}D_{13}&D_{23}&D_{33}&\cdots&D_{N/4,3}\\
	    \vdots&\vdots&\vdots&\ddots&\vdots\\
		\frac{1}{2}D_{1,N/4-1}&D_{2,N/4-1}&D_{3,N/4-1}&\cdots&D_{N/4,N/4-1}\\
		\frac{1}{4}D_{1,N/4}&\frac{1}{2}D_{2,N/4}&\frac{1}{2}D_{3,N/4}&\cdots&\frac{1}{2}D_{N/4,N/4}\\
	\end{array}
	\right]
    \left[
   \begin{array}{c}
    H_1\\
    H_2\\
    \vdots\\
    H_{N/4}
   \end{array}
  \right].\nonumber\\
\end{eqnarray}
Therefore, we obtain
\begin{eqnarray}
	b_{2s} &=& \frac{8G_1}{N}\left[\frac{1}{2}D_{1s}H_1 + \sum_{k=2}^{N/4}{D_{ks}H_{k}}\right]\nonumber\\
    &=& \frac{8G_1}{N}\left[\frac{1}{2}H_1 + \sum_{k=2}^{N/4}{\cos{\frac{4(k-1)s\pi}{N}} H_k}\right]\nonumber\\
	&=&\frac{24G_1}{(N^2-4)\sin^2{(2s\pi/N)}}\quad(s=1,2,\cdots, N/4-1)\nonumber\\
	b_{N/2}&=&\frac{12G_1}{N^2-4}
	\label{eqn:explicit_sol_be1}
\end{eqnarray}
From the relation (\ref{eqn:g1b1}), we obtain
\begin{eqnarray}
	w_{2s} &=& \frac{b_{2s}}{b_1} = \frac{\sin^2{(\pi/N)}}{\sin^2{(2s\pi/N)}}\quad(s=1,2,\cdots,N/4-1)\nonumber\\
	w_{N/2}&=&\frac{b_{N/2}}{b_1}=\frac{\sin^2{(\pi/N)}}{2}\quad(s=N/4)
	\label{sol:w_even}
\end{eqnarray}
Combining (\ref{sol:w_odd}) and (\ref{sol:w_even}), the explicit solution is
\begin{eqnarray}
	b_{r} &=& b_1\frac{\sin^2{(\pi/N)}}{\sin^2{(r\pi/N)}}\quad(r=1,2,\cdots,N/2-1)\nonumber\\
	b_{N/2} &=& \frac{b_1}{2}\sin^2{(\pi/N)}\quad(r=N/2).
\end{eqnarray}

\clearpage

\end{document}